\documentclass[11pt,letterpaper]{amsart}
\usepackage[foot]{amsaddr}
\usepackage[margin=1.25in]{geometry}
\usepackage[square,numbers,sort&compress]{natbib}

\usepackage{subcaption}
\usepackage{amsmath,amsthm}
\usepackage{tikz}
\usetikzlibrary{arrows,automata} 
\usepackage{color}
\usepackage{graphicx}
\usepackage{mathrsfs}
\usepackage{url}
\usepackage{enumerate}
\usepackage{comment}

\usepackage{wasysym}
\usepackage{algorithm}
\usepackage[hidelinks]{hyperref}
\usepackage{cleveref}
\usepackage[noend]{algpseudocode}
\usepackage{tikz}
\algrenewcommand\algorithmicrequire{\textbf{Input:}}
\algrenewcommand\algorithmicensure{\textbf{Output:}}
\algnewcommand\algorithmicforeach{\textbf{for each}}
\algdef{S}[FOR]{ForEach}[1]{\algorithmicforeach\ #1\ \algorithmicdo}

\newcommand{\eps}{\varepsilon}
\newcommand{\OPT}{S^*}
\newcommand{\MG}{\mathrm{\textsc{MG}}}
\newcommand{\AG}{\mathrm{\textsc{AG}}}
\newcommand{\N}{N}
\newcommand{\capacity}{B}

\colorlet{mycolor1}{violet!80!black}
\colorlet{mycolor2}{orange!70!black}
\colorlet{mycolor3}{green!70!black}

\allowdisplaybreaks

\theoremstyle{plain}
\newtheorem{theorem}{Theorem}[section]
\newtheorem{lemma}[theorem]{Lemma}
\newtheorem{proposition}[theorem]{Proposition}

\newtheorem*{theorem*}{Claim}

\theoremstyle{definition}
\newtheorem{definition}[theorem]{Definition}

\vfuzz \hfuzz
\begin{document}

\title[Maximizing a Submodular Function under an Unknown Knapsack Constraint]{Maximizing a Submodular Function with Bounded Curvature under an Unknown Knapsack Constraint} 

\author{Max Klimm}
\author{Martin Knaack}


\begin{abstract}
	This paper studies the problem of maximizing a monotone submodular function under an unknown knapsack constraint. A solution to this problem is a policy that decides which item to pack next based on the past packing history. The robustness factor of a policy is the worst case ratio of the solution obtained by following the policy and an optimal solution that knows the knapsack capacity. We develop a policy with a robustness factor that is decreasing in the curvature $c$ of the submodular function. For the extreme cases $c=0$ corresponding to an additive objective function, it matches a previously known and best possible robustness factor of $1/2$. For the other extreme case of $c=1$ it yields a robustness factor of~$\approx 0.35$ improving over the best previously known robustness factor of~$\approx 0.06$. 

    The analysis of our policy relies on a greedy algorithm that is a slight modification of Wolsey's greedy algorithm for the submodular knapsack problem with a known knapsack constraint. We obtain tight approximation guarantees for both of these algorithms in the setting of a submodular objective function with curvature $c$.
\end{abstract}

\maketitle

\newpage
 
\section{Introduction}
\label{sec:introduction}

This paper is concerned with the problem
\begin{align}
    \label{eq:problem}
    \text{maximize } \Biggl\{ f(S) \;\Bigg\vert\; S \subseteq N \text{ and } \sum_{i \in S} s(i) \leq \capacity \Biggr\}
\end{align}
of maximizing a submodular, monotone, and normalized function $f \colon 2^N \to \mathbb{R}_{\geq 0}$ under a knapsack constraint, where $N$ is a finite set of items, $s(i) \in \mathbb{R}_{>0}$ is the size of item~$i \in N$, and $\capacity \in \mathbb{R}_{>0}$ is a knapsack capacity. 
This optimization problem is an important abstraction of many problems that appear in various applications, such as facility location (Cornu\'{e}jols et al.~\cite{OTHER:cornuejols77}), sensor placement (Krause and Guestrin~\cite{DBLP:journals/tist/KrauseG11}, Krause et al.~\cite{DBLP:journals/jmlr/KrauseSG08}), marketing in social networks (Kempe et al.~\cite{DBLP:journals/toc/KempeKT15}), and maximum entropy sampling (Lee~\cite{DBLP:journals/ior/Lee98a}). 

For the special case of a cardinality constraint where $s(i) = 1$ for all $i \in N$, a straightforward greedy algorithm by Nemhauser et al.~\cite{DBLP:journals/mor/NemhauserW78} computes a solution with an approximation guarantee of $1 - 1/e$ and this ratio is best possible for any polynomial algorithm unless $\mathsf{P} = \mathsf{NP}$ (Feige~\cite{DBLP:journals/jacm/Feige98}). For the case of a general knapsack constraint, combining the greedy algorithm with a partial enumeration of all subsolutions with at most three items yields the same approximation guarantee (Sviridenko~\cite{DBLP:journals/orl/Sviridenko04}).

While these results are tight, the algorithms often perform much better than their theoretical guarantees. In order to explain and quantify this phenomenon, Conforti and Cornu\'{e}jols~\cite{DBLP:journals/dam/ConfortiC84} introduce the concept of the \emph{curvature} of a submodular function. Recall that a function $f \colon 2^N \to \mathbb{R}_{\geq 0}$ is submodular if the marginal increase $f(S \cup \{u\}) - f(S)$ of an element $u \in N \setminus S$ is non-increasing as $S$ increases. The curvature~$c \in [0,1]$ measures how much this marginal increase of an item~$u$ varies when varying $S$ and is defined as
\begin{equation*}
    c = 1 - \min_{j \in \N} \frac{f(\N) - f(\N \setminus \{j\})}{f(\{j\})},
\end{equation*}
where we further used that $f$ is normalized, i.e., $f(\emptyset) = 0$. It is easy to see that $c = 0$ if and only if the function is additive. The other extreme case $c = 1$ is, e.g., attained when $f$ is the rank function of a matroid. Conforti and Cornu\'{e}jols~\cite{DBLP:journals/dam/ConfortiC84} show that the greedy algorithm for the cardinality constraint case has an improved approximation guarantee of $(1-e^{-c})/c$. A more sophisticated algorithm for the same problem by Sviridenko et al.~\cite{DBLP:journals/mor/SviridenkoVW17} achieves an even better approximation guarantee of $1 - c/e - \eps$ for any $\eps > 0$. 

In all of the results above it is assumed that all data of the problem~\eqref{eq:problem} is given completely. In this paper, we consider a variant of the problem where the set of items $N$, their sizes $s(i) \in \mathbb{R}_{>0}$, and the function $f \colon 2^N \to \mathbb{R}_{\geq 0}$ are known, but the knapsack capacity $\capacity \in \mathbb{R}_{>0}$ is unknown. In this context, a solution to the problem is a policy that decides which item to pack next, based on the previous packing history. More formally, a policy $\Pi$ is a binary decision tree where nodes correspond to items with the property that no item appears more than once on a path from the root to a leaf. The item at the root of the tree is the item that is attempted to be packed first. If it fits, it is irrevocably included in the solution, the (unknown) capacity is reduced by the size of the item, and the solution proceeds with the left subtree of the decision tree. If the item does not fit, it is discarded, the (unknown) capacity stays the same, and the solution proceeds with the right subtree. This process stops after a leaf is reached.
The assumption that the policy can resume packing smaller items after a larger item does not fit is suitable when the knapsack capacity is interpreted as a monetary budget.
Generally speaking, such packing policies are desirable when packing problems of this kind have to be solved repeatedly for varying knapsack capacities. For illustration, consider the marketing problem in social networks. By analyzing the social network, a packing policy can be constructed that can then be used in order to run marketing campaigns \emph{for all possible budgets}, without the need to rerun any optimization. In a similar vein, consider the problem of maximum entropy sampling. The Shannon entropy of a set of (dependent) random variables is a submodular function of the (index set) of the variables. Suppose that observing the realization of a random variable comes at a cost (for market research, for evaluating the data, etc.). With our algorithm, one can compute a policy that \emph{for all budgets} allows to retrieve close to optimal information without any need to rerun the optimization for different budgets. 

In the examples above, we clearly want to obtain solutions that are good for any possible capacity. We evaluate the quality of a policy in terms of its \emph{robustness factor}.
Fix an instance of \eqref{eq:problem}, and a corresponding policy $\Pi$.
For a capacity $\capacity \in \mathbb{R}_{>0}$, let $\Pi(\capacity)$ be the set of items packed by the policy when the knapsack capacity is $\capacity$, and let $\OPT(\capacity)$ be the items included in an optimal solution for capacity~$\capacity$. The robustness factor is defined as $\alpha = \inf_\capacity f(\Pi(\capacity)) / f(\OPT(\capacity))$. A policy with robustness factor of $\alpha \in [0,1]$ is called $\alpha$-optimal.

\subsection{Our Results and Techniques}

For the case that $f$ is additive (corresponding to the case that $c = 0$), Disser et al.~\cite{DBLP:journals/siamdm/DisserKMS17} show that every instance admits a $1/2$-optimal policy, and that the factor of $1/2$ is best possible.
Kawase et al.~\cite{DBLP:journals/siamdm/KawaseSF19} consider the fully submodular case corresponding to the case $c=1$. They provide a deterministic policy with robustness factor $2(1-1/e)/21 \approx 0.06$ and a randomized policy with robustness factor of $(1- 1/e)/2 \approx 0.32$.

We provide a deterministic polynomial algorithm that constructs a deterministic policy $\Pi$ with a robustness factor of
\begin{align}
    \alpha = \frac{1}{c} \bigl(1- e^{-cx}\bigr), \label{eq:results}
\end{align}
where $x$ is the unique root of the equation
$\smash{\frac{1}{c} \bigl(1- e^{-c z} \bigr) = \frac{1-z}{1+c(1-z)}}$.

For the most general case of a submodular function with curvature $c=1$, this yields a robustness factor of $\approx 0.35$ which improves over the factor of $\approx 0.06$ by Kawase et al.; for smaller values of $c < 1$ the robustness factor increases and retains the optimal factor of $\alpha =1/2$ for the additive case when $c=0$. For an illustration; see Figure~\ref{fig:robustness-factor}.

\begin{figure}
    \begin{center}
        \begin{tikzpicture}[scale=3.5]
            \draw[fill=lightgray,draw=none] (0,0.5) -- (1,0.5) -- (1,1) -- (0,1) -- cycle; 
            \node (imp) at (0.5,0.6) {impossible~\cite{DBLP:journals/siamdm/DisserKMS17}}; 
            \draw[-latex,very thick] (-0.1,0) to (1.1,0) node[right] {$c$};
            \draw[-latex,very thick] (0,-0.1) to (0,1.1) node[above] {$\alpha$};
            \draw[thick] plot[smooth] coordinates {(0,0.5) (0.1, 0.4818) (0.2, 0.4647) (0.3, 0.4485) (0.4, 0.4333) (0.5, 0.4189) (0.6, 0.4053) (0.7, 0.3925) (0.8, 0.3803) (0.9, 0.3687) (1, 0.3578)};
            \node (this) at (0.35,0.3) {\raisebox{5pt}{this work$\nearrow$}};
            \node[circle, minimum size=4pt, inner sep=0pt, fill=black, label=below left:\raisebox{-5pt}{Disser et al.~\cite{DBLP:journals/siamdm/DisserKMS17}} 
            $ \nearrow $] (disser) at (0,0.5) {};
            \node[circle,minimum size=0pt,inner sep=0pt,fill=black,label=above right:$\swarrow$\raisebox{5pt}{Kawase et al.~\cite{DBLP:journals/siamdm/KawaseSF19}}] (disser) at (0.7,0.0602) {};
            \draw[thick] (0,0.0602) -- (1,0.0602);
            \draw (1,0.03) to (1,-0.03) node[below] {$1$};
            \draw (0.5,0.03) to (0.5,-0.03) node[below] {$1/2$};
            \draw (0.03,0.5) to (-0.03,0.5) node[left] {$1/2$};
            \draw (0.03,1) to (-0.03,1) node[left] {$1$};
        \end{tikzpicture}
        \caption{
            \label{fig:robustness-factor}
            Robustness factors $\alpha$ of deterministic policies as a function of the curvature~$c$ achieved by this and previous work.}
    \end{center}
\end{figure}

A central technique for solving submodular maximization problems with a known or unknown knapsack capacity are greedy algorithms, and this work is no exception. Disser et al.~\cite{DBLP:journals/siamdm/DisserKMS17} compare the solution obtained by their policy with a greedy algorithm called \textsc{MGreedy} that either takes the greedy sequence or the first item that does not fit the knapsack anymore. As discussed by Kawase et al.~\cite{DBLP:journals/siamdm/KawaseSF19} this approach seems difficult to apply to submodular functions because the greedy sequence is different for different sizes of the knapsack due to the substitute effects among the items for the objective. They instead single out valuable items that provide a significant ratio of the optimum solution. This approach, however, comes at the expense of a much lower robustness factor. 

Actually, we can show that it is impossible to find deterministic policies that are always at least as good as \textsc{MGreedy} in the submodular setting. We circumvent this issue by analyzing a different kind of greedy algorithm that we call \textsc{AGreedy} and that seems to be more compatible with robust policies. In Section~\ref{sec:known}, we provide a tight analysis for this algorithm for the case of a known knapsack capacity and show that its approximation guarantee is the factor given in \eqref{eq:results}. As a byproduct of our analysis, we further obtain that the \textsc{MGreedy} algorithm also has the approximation guarantee as in \eqref{eq:results}. This generalizes a result of Wolsey~\cite{DBLP:journals/mor/Wolsey82} who analyzed this algorithm only for the general submodular case where $c=1$. In Section~\ref{sec:robust}, we then devise an adaptive policy that achieves a robustness that is at least as good as the approximation guarantee of \textsc{AGreedy}. 

\subsection{Further Related Work}

The problem of maximizing a submodular function under different constraints has a long history in the optimization literature. Nemhauser et al.~\cite{DBLP:journals/mp/NemhauserWF78} consider the problem of maximizing a monotonic submodular function under a cardinality constraint and show that the greedy algorithm that iteratively adds an item that maximizes the increase of the objective function achieves an approximation guarantee of  $(1-1/e) \approx 0.63$.
Nemhauser and Wolsey~\cite{DBLP:journals/mor/NemhauserW78} prove that this ratio is best possible for algorithms that have only access to $f$ via a value oracle that can only be queried a polynomial number of times.
Even for the special case that $f$ is given explicitly and corresponds to a maximum coverage function, there is no better approximation possible in polynomial time, unless $\mathsf{P} = \mathsf{NP}$, as shown by Feige~\cite{DBLP:journals/jacm/Feige98}.
Wolsey~\cite{DBLP:journals/mor/Wolsey82} considers the more general problem of maximizing a submodular function under a knapsack constraint and achieves an approximation guarantee of $1-e^{-x} \approx 0.35$ where $x$ is the unique root of the equation $e^x = 2-x$. Sviridenko~\cite{DBLP:journals/orl/Sviridenko04} shows that a combination of the greedy algorithm with a partial enumeration scheme achieves an approximation guarantee of $1-1/e$.
Another way to generalize the cardinality constrained case is to allow for arbitrary matroid constraints. For this case, the greedy algorithm yields an approximation guarantee of $1/2$, as shown by Fisher et al.~\cite{OTHER:fisherNW78}. Calinescu et al.~\cite{DBLP:journals/siamcomp/CalinescuCPV11} achieve a $1-1/e$ approximation by solving a fractional relaxation of the problem and combining it with a suitable rounding technique.

Conforti and Cornu\'ejols~\cite{DBLP:journals/dam/ConfortiC84} introduce the curvature~$c$ as a measure for the non-linearity of a (submodular) function and show that the greedy algorithm has an approximation guarantee of $(1- e^{-c})/c$ for the case of a cardinality constraint and $1/(c+1)$ for the case of a matroid constraint.
Vondr{\'a}k~\cite{OTHER:vondrak10} shows that the continuous greedy algorithm yields an approximation guarantee of $(1-e^{-c})/c$ for the case of a matroid constraint, and proves that no better approximation is possible in the value oracle model with a polynomial number of queries.
Sviridenko et al.~\cite{DBLP:journals/mor/SviridenkoVW17} give an algorithm with approximation guarantee of $1 - c/e - \mathcal{O}(\eps)$ for the problem with a matroid constraint. Yoshida~\cite{DBLP:journals/siamdm/Yoshida19} obtains the same approximation guarantee for the problem under a knapsack constraint. The algorithm relies on a continuous version of the greedy algorithm which seems to be incompatible with an unknown knapsack constraint since many items will be fractional during the course of the algorithm for smaller knapsack constraints. Also the distinction between small and large items which is elementary in the algorithm cannot be employed when the capacity is not known.

Packing problems with an unknown knapsack are studied by Megow and Mestre~\cite{DBLP:conf/innovations/MegowM13}. They consider the additive case and assume that the policy stops when an item does not fit the knapsack.
In this setting, no constant robustness factor is achievable on all instances and Megow and Mestre provide a polynomial time approximation scheme (PTAS) for the computation of an optimal policy. Navarra and Pinotti~\cite{DBLP:journals/tcs/NavarraP17} show how to construct a policy with robustness factor $1/2$ for instances that have the property that every item fits into the empty knapsack.
Disser et al.~\cite{DBLP:conf/waoa/DisserKW21} consider the optimization of a fractionally subadditive objective with the additional property that every singleton set has a value of $1$, and give a policy with robustness factor of $\approx 0.30$.
For the case of an unknown cardinality constraint, there is no difference between policies that continue or stop packing after an item does not fit. Bernstein et al.~\cite{OTHER:journals/mp/Bernsetin20} introduce a property on the objective function that they term accountability and that is more general than submodularity. They show that the optimal robustness factor for maximization of an accountable objective under an unknown cardinality constraint is between $1/(1+\phi) \approx 0.38$ where $\phi$ is the golden ratio and $0.46$.

\section{Preliminaries}

\subsection{Submodular Functions}
Let $\N$ be a finite set. A function $f \colon 2^\N \to \mathbb{R}_{\geq 0}$ is called \emph{monotone} if $ f(S) \leq f(T) $ for every $ S,T \in 2^{\N}$ with $S \subseteq T$, is called \emph{normalized} if $ f(\emptyset) = 0$, and is called \emph{submodular} if 
$f(S) + f(T) \geq f(S \cup T) + f(S \cap T)$ for all $S, T \in 2^{\N}$. 
For our purposes, it is without loss of generality to assume that $f(\{j\}) > 0$ for all $j \in N$ since an element $j$ with $f(\{j\}) = f(\emptyset)$, by submodularity, has no influence on the value of $f$ and, thus, can be removed from \eqref{eq:problem}.
As a shorthand, we use $f(u) = f(\{u\})$ for a single element $u \in N$ and $f(u \mid S) = f(S \cup \{u\}) - f(S) $ for the marginal increase of $u \in N$ with respect to a set $S \in 2^{\N}$. It is well-known that a function $f$ is submodular if and only if the following statement is satisfied:
\begin{align}
    f(u \mid S) \geq f(u \mid T) \qquad  \text{for all $S \subseteq T \subseteq \N$, $u \in \N \setminus T$}. \label{def:Subm1}
\end{align}
A submodular and monotone function further satisfies the following inequality, see, e.g., Nemhauser et al.~\cite{DBLP:journals/mp/NemhauserWF78} for a reference 
\begin{equation}
    \label{eq:nemhauser}
    f(T) \leq f(S) + \sum_{u \in T \setminus S} f(u \mid S) \qquad \text{for all $S \subseteq T \subseteq \N$}.
\end{equation}

\subsection{Curvature}
The curvature of a normalized, monotone and submodular function $f \colon 2^N \to \mathbb{R}_{\geq 0}$ is defined as
\begin{align*}
    c = 1 - \min_{j \in N} \frac{f(j \mid N \setminus \{j\})}{f(j)}.	
\end{align*}
The following lemma summarizes a couple of inequalities that are valid for submodular functions with a given curvature that are easy to show yet useful for the remainder of the paper. 	
\begin{lemma}
    \label{lem:curv-lemma}
    For a normalized, monotonic, and submodular function $ f \colon 2^\N \to \mathbb{R}_{\geq 0} $ with curvature $ c \in [0,1]$, the following inequalities are satisfied:
    \begin{enumerate}[(i)]
        \item \label{it:curv-lemma-1} $f(j \mid S) \geq (1 - c) \, f(j) $ for all $S \subset N$ and all $j \in N \setminus S$;
        \item \label{it:curv-lemma-2} $ f(S \cup T) \geq f(S) + (1-c) \sum_{i \in T} f(i)$  for all $S, T \subset \N$  with $S \cap T = \emptyset$.
    \end{enumerate}
\end{lemma}

\begin{proof}
    We first show \eqref{it:curv-lemma-1}.
    Let $S \subset N$ and $j \in N \setminus S$ be arbitrary. We calculate
    \begin{align*}
        1 - c =  \min_{i \in N} \frac{f(i \mid N \setminus \{i\})}{f(i)} \leq \frac{f(j \mid N \setminus \{j\})}{f(j)} \leq \frac{f(j \mid S)}{f(j)}
    \end{align*}
    where we first used the definition of curvature and at the end, we applied \eqref{def:Subm1}.
    
    To show \eqref{it:curv-lemma-2}, we successively apply \eqref{it:curv-lemma-1} on the elements in $T$.
\end{proof}

\subsection{Submodular Maximization under a Knapsack Constraint}
An instance of the submodular maximization problem under a \emph{known} knapsack constraint is given by a set of $n$ items $N = \{i_1,i_2,\dots,i_n\}$ where each item $i \in N$ has a size $s(i) \in \mathbb{R}_{> 0}$. We are further given a monotone, normalized and submodular function $f \colon 2^N \to \mathbb{R}_{\geq 0}$ that assigns a value $f(S)$ to every subset $S \subseteq N$ of items, and a capacity $\capacity \in \mathbb{R}_{>0}$.
For a subset $T \subseteq N$, we write $s(S) = \sum_{i \in S} s(i)$.
A solution to the problem is a set of items $S \subseteq N$. A solution $S$ is called \emph{feasible} if $s(S) \leq \capacity$, and called \emph{optimal} if $f(S) \geq f(T)$ for every feasible solution~$T$.

An instance of the submodular maximization problem under an \emph{unknown} knapsack constraint is as above except that we do not know the capacity $\capacity \in \mathbb{R}_{>0}$, i.e., we are again given a set of items $N$, their sizes $s(i)$, $i \in N$ and the submodular function $f$. A solution to this problem  is a policy $\Pi$ that governs the order in which items are added to the solution.

\section{Submodular Knapsack Problem with Known Capacity}
\label{sec:known}

\begin{figure}[t]
    \begin{minipage}{0.48\textwidth}
        \begin{algorithm}[H]
            \caption{Modified Greedy Algorithm \textsc{MGreedy}}
            \begin{algorithmic}
                \State $G_0 \gets \emptyset; j \gets 1$
                \State $U \gets \{i \in N \mid s(i) \leq \capacity\}$
                \While{$ U \neq \emptyset $}
                \State $ i_j \gets \arg\max_{i \in U} \Bigl\{\frac{f(i \,|\, G_{j-1})}{s(i)}\Bigr\}$
                \If{$ s(G_{j-1} \cup \{i_j\}) \leq \capacity $}
                \State $ G_j \gets G_{j-1} \cup \{i_j\} $
                \State $ U \gets U \setminus \{i_j\} $
                \State $j \gets j+1$
                \Else
                \State \textbf{break}
                \EndIf
                \EndWhile
                \State $k \gets j-1$
                \If{$ U = \emptyset $}
                \State \Return $ G_k $
                \Else
                \If{$ f(G_k) \geq f(i_{k+1}) $}
                \State \Return $ G_k $
                \Else
                \State \Return $ \{i_{k+1}\} $ 
                \EndIf
                \EndIf
            \end{algorithmic}
        \end{algorithm}
    \end{minipage}
    \hfill
    \begin{minipage}{0.48\textwidth}
        \begin{algorithm}[H]
            \caption{Alternative Greedy Algorithm \textsc{AGreedy}}
            \begin{algorithmic}
                \State $ G_0 \gets \emptyset; j \gets 1$
                \State $U \gets \{i \in N \mid s(i) \leq \capacity\}$
                \While{$ U \neq \emptyset $}
                \State $ i_j \gets \arg\max_{i \in U} \Bigl\{\frac{f(i \,|\, G_{j-1})}{s(i)}\Bigr\} $
                \If{$ s(G_{j-1} \cup \{i_j\}) \leq \capacity $}
                \State $ G_j \gets G_{j-1} \cup \{i_j\} $
                \State $ U \gets U \setminus \{i_j\} $
                \State $j \gets j+1$
                \Else
                \State \textbf{break}
                \EndIf
                \EndWhile
                \State $k \gets j-1$
                \If{$ U = \emptyset $}
                \State \Return $G_k$ 
                \Else
                \If{$ f(G_k) \geq f(i_{k+1} \mid G_k) $}
                \State \Return $G_k$ 
                \Else
                \State \Return $ \{i_{k+1}\} $
                \EndIf
                \EndIf
            \end{algorithmic}
        \end{algorithm}
    \end{minipage}
    \caption{Greedy algorithms for maximizing a submodular function $f$ over a knapsack constraint.}
\end{figure}

In this section, we analyze the approximation guarantee for two natural greedy algorithms that, for the sake of a better distinction, we call \emph{modified greedy algorithm} (\textsc{MGreedy}) and \emph{alternative greedy algorithm} (\textsc{AGreedy}).
The modified greedy algorithm was proposed and analyzed by Wolsey~\cite{DBLP:journals/mor/Wolsey82} where he shows that it has an approximation ratio of $1-e^{-x} \approx 0.35$ where $x$ is the unique root of the equation $e^x = 2-x$.
To the best of our knowledge, there is no better analysis of this algorithm for submodular functions with bounded curvature.
The alternative greedy algorithm is a slight variation of this algorithm that we need in order to derive policies for the optimization problem with unknown knapsack constraints in Section~\ref{sec:robust}.

Both algorithms first discard all items~$i$ that do not fit into an empty knapsack, i.e., where $s(i) > \capacity$. Then, the algorithms start in iteration~$0$ with an empty solution $G_0 = \emptyset$. In every iteration~$j = 1,2,\dots$, both algorithms choose an item
\begin{equation*}
    i_j \in \arg\max \biggl\{\frac{f(i \mid G_{j-1})}{s(i)} \;\bigg\vert\; i \in N \setminus G_{j-1} \biggr\}
\end{equation*}
that is not yet contained in the solution $G_{j-1}$ and maximizes the ratio of the increment of the objective function and the size of the item. If item~$i_j$ still fits the knapsack, i.e., $s(G_{j-1} \cup \{i_j\}) \leq \capacity$, then the item is added to the solution. Otherwise, the algorithm stops. Let $k$ be the last index such that item $i_k$ still fits into the knapsack.

Then, algorithm $\textsc{MGreedy}$ either returns the better of the solutions $G_k$ and $\{i_{k+1}\}$, i.e., it either returns the maximum prefix of the greedy sequence that still fits into the knapsack, or the first item that did not fit into the knapsack anymore. The alternative greedy also either returns $G_k$ or $\{i_{k+1}\}$ but the rule when to return one of the solutions slightly differs. The item $\{i_{k+1}\}$ is only returned if the marginal increase $f(i_{k+1} \mid G_k)$ of adding it to $G_k$ is larger than $f(G_k)$. In all other cases, $G_k$ is returned.

Since \textsc{MGreedy} always returns the better of the two solutions $G_k$ and $\{i_{k+1}\}$ while $\textsc{AGreedy}$ may also return $G_k$ even though $f(G_k) < f(i_{k+1})$, it is clear that the solution returned by \textsc{MGreedy} is always at least as good as the one returned by \textsc{AGreedy}.
Thus, the following result is immediate. 

\begin{proposition}
    \label{prop:greedy-algorithms}
    For every instance, $f(\MG) \geq f(\AG)$. 
\end{proposition}

Despite this fact, we are still interested in analyzing \textsc{AGreedy} for two reasons. First, it turns out that \textsc{AGreedy} is better suited in order to design robust packing policies for the problem with an unknown knapsack capacity. Second, it will turn out, that in the worst case, the approximation guarantees that we obtain for \textsc{MGreedy} and \textsc{AGreedy} are actually the same.

In the following, we fix an instance of the submodular maximization problem under a knapsack constraint with known capacity. We assume that the items $N = \{i_1,i_2,\dots,i_n\}$ are ordered in the order as they would be considered by the greedy algorithms and we call this order the \emph{greedy order} of $N$. We also set $s_j = s(i_j)$ for all $j \in \{1,\dots,n\}$ and further, we let $k$ be the maximal prefix of this ordering that still fits into the knapsack, i.e., $\smash{k = \max \bigl\{j \in \{1,\dots,n\}\mid \sum_{i=1}^j s_i \leq \capacity \bigr\}}$. We let $\OPT$ denote the set of items in an optimal solution and we let $\MG$ and $\AG$ denote the set of items returned by \textsc{MGreedy} and \textsc{AGreedy}, respectively. Further, for $j \in \{0,\dots,n\}$, we let $\smash{G_j = \bigcup_{l=1,\dots,j} \{i_l\}}$ denote the first $j$ items of the greedy order.

Before we start with the analysis of the greedy algorithms, we want to give some intuition about the cases where the curvature lies in $(0,1)$ versus the case where $c = 1$. Wolsey's proof for \textsc{MGreedy} in the latter case consists of three steps: 
(i) bound the marginal increase $f(i_j \mid G_{j-1})$ for each iteration $j \in \{1,\dotsc,k+1\}$;
(ii) use the first $k$ bounds from step~(i) inductively, to derive a lower bound for the value of $G_k$;
(iii) combine the bound from step~(ii) with the bound for the marginal increase of the item in iteration $k+1$ from step~(i), to obtain a bound for $\max\{f(G_k),f(i_{k+1})\}$.

The first step is where we can improve the analysis for $c \in (0,1)$. Actually, the bound does not only hold for $ f(\OPT) - f(G_{j-1}) $, but also for $ f(\OPT \cup G_{j-1}) - f(G_{j-1}) $ and here we can apply part~(\ref{it:curv-lemma-2}) of \Cref{lem:curv-lemma} to derive the stronger bound $ f(\OPT) - f(G_{j-1}) + \sum_{i \in G_{j-1} \setminus \OPT} f(i) $. 
Note that the bound is only stronger if we packed items that are not in $\OPT$ in previous iterations. Intuitively, one would assume that packing items from an optimal solution cannot be harmful to the approximation guarantee and in the end it also turns out that this is the case. Nevertheless, we have to keep track of the iterations where we added items from $\OPT$ throughout all subsequent iterations in the proof. In order to do this, we introduce some more notation. 
For iteration $j \in \{0,\dotsc,k\}$, we let $ Q_j = \{\ell \in \{1,\dotsc,j\} \mid i_\ell \in \OPT \} $ be the index set of the first $j$ items of the greedy order that are in $\OPT$ and we define
\begin{equation*}
    \delta_j = \frac{f(\OPT)-\sum_{\ell \in Q_j} f(i_\ell \mid G_{\ell-1})}{\capacity - s(\OPT \cap G_j)}.
\end{equation*}
Note that $ i_j \notin \OPT$ implies $\delta_j = \delta_{j-1}$ and therefore, $\delta_j \neq \delta_{j-1}$ implies $ i_j \in \OPT$. In order to avoid division by zero, we exclude the trivial case $ G_k = \OPT $ for the rest of the chapter such that $ \capacity - s(\OPT \cap G_j) > 0 $ for all $ j \in \{1,\dots,k\} $. 
Actually, we are interested in the iterations $\ell \in \{1,\dotsc,k\}$ where $i_\ell \in \OPT$ and $\delta_\ell < \delta_{\ell-1}$. In these cases, we derive an additional term that we shorthand by $\omega_{\ell,\ell}$ and, for a later iteration $ j \in \{\ell+1,\dotsc,k\}$, we let $\omega_{\ell,j}$ denote the remains of that term from iteration $\ell$. 
Formally, for $\ell \in \{1,\dotsc,k\}$ and $j \in \{\ell,\dotsc,k\}$, we define
\begin{equation*}
    \omega_{\ell,j} = 
    \begin{cases}
        \Biggl(\displaystyle\prod_{m=\ell+1}^{j} \biggl(1-\frac{c \, s_m}{\capacity}\biggr)\Biggr) (\capacity - s(G_j)) (\delta_{\ell-1} - \delta_\ell)  &\text{ if } \delta_\ell < \delta_{\ell-1} , \\
        0  &\text{ otherwise},
    \end{cases}
\end{equation*}
and additionally, we set $\omega_{k+1,k+1} = 0 $.

Now we can start with the proof. We are following the afore-mentioned steps (i), (ii) and (iii) in \Cref{lem:greedy-increase}, \Cref{lem:greedy-sum} and \Cref{theo:greedy-approximation}, respectively.
Thus, in \Cref{lem:greedy-increase}, we bound the marginal increase of the greedy solution in each iteration from below.

\begin{lemma}
    \label{lem:greedy-increase}
    For each $ j \in \{1,\dotsc,k+1\} $, we have
    \begin{equation*}
        f(i_j \mid G_{j-1}) \geq \frac{c \, s_j}{\capacity} \Bigl(f(\OPT) - f(G_{j-1})\Bigr) + (1-c) s_j \delta_{j-1} + \omega_{j,j}.
    \end{equation*}
\end{lemma}

\begin{proof}
    Let $j \in \{1,\dotsc,k+1\}$ be arbitrary. We first show that the statement always holds for $\omega_{j,j} = 0$. We use part~\eqref{it:curv-lemma-2} of \Cref{lem:curv-lemma} with $ S = \OPT $ and $ T = G_{j-1} \setminus \OPT $. Subtracting $ f(G_{j-1}) $ from both sides, yields
    \begin{equation}
        \label{eq:lem-greedy-increase}
        f(\OPT) + (1-c) \sum_{i \in G_{j-1} \setminus \OPT} f(i) - f(G_{j-1}) \leq f(\OPT \cup G_{j-1}) - f(G_{j-1}). 
    \end{equation}
    We can bound $f(\OPT \cup G_{j-1}) - f(G_{j-1})$ on the right hand side by \eqref{eq:nemhauser} and by the definition of the greedy sequence. We obtain
    \begin{align*}
    \begin{split}
        f(\OPT \cup G_{j-1}) - f(G_{j-1}) &\leq \sum_{i \in \OPT \setminus G_{j-1}} f(i \mid G_{j-1}) \\
        & = \sum_{i \in \OPT \setminus G_{j-1}} s(i) \, \frac{f(i \mid G_{j-1})}{s(i)} \\
        & \leq \Biggl(\sum_{i \in \OPT \setminus G_{j-1}} s(i)\Biggr) \frac{f(i_j \mid G_{j-1})}{s_j} \\
        & = s(\OPT \setminus G_{j-1}) \frac{f(i_j \mid G_{j-1})}{s_j}.
    \end{split}
    \end{align*}
    Moreover, we can bound the left hand side of \eqref{eq:lem-greedy-increase}. For $j \in \{1,\dotsc,k\}$, we let $\overline{Q}_j = \{1,\dotsc,j\} \setminus Q_j$ be the complement of the index set $Q_{j}$. We have 
    \begin{equation*}
        \sum_{i \in G_{j-1} \setminus \OPT} f(i) = \sum_{\ell \in \overline{Q}_{j-1}} f(i_\ell) \geq \sum_{\ell \in \overline{Q}_{j-1}} f(i_\ell \mid G_{\ell-1})
    \end{equation*}
    by submodularity of $f$ and together with $ f(G_{j-1}) = \sum_{\ell \in [j-1]} f(i_\ell \mid G_{\ell-1}) $, we get
    \begin{multline*}
        f(\OPT) + (1-c) \sum_{i \in G_{j-1} \setminus \OPT} f(i) - f(G_{j-1}) \\
        \begin{aligned}
        &\geq f(\OPT) + (1-c) \sum_{\ell \in \overline{Q}_{j-1}} f(i_\ell \mid G_{\ell-1}) - \sum_{\ell = 1}^{j-1} f(i_\ell \mid G_{\ell-1}) \\
        &= f(\OPT) -(1-c) \sum_{\ell \in Q_{j-1}} f(i_\ell \mid G_{\ell-1}) - c \sum_{\ell = 1}^{j-1}f(i_\ell \mid G_{\ell-1}) \\
        &= c \Bigl( f(\OPT) - f(G_{j-1})\Bigr)  + (1-c) \Bigl( f(\OPT) - \sum_{\ell \in Q_{j-1}} f(i_\ell \mid G_{\ell-1})\Bigr).
        \end{aligned}
    \end{multline*}
    Combining both bounds with \eqref{eq:lem-greedy-increase} yields
    \begin{equation*}
        s(\OPT \setminus G_{j-1}) \frac{f(i_j \mid G_{j-1})}{s_j} \geq c \Bigl( f(\OPT) - f(G_{j-1})\Bigr)  + (1-c) \Bigl( f(\OPT) - \sum_{\ell \in Q_{j-1}} f(i_\ell \mid G_{\ell-1})\Bigr).
    \end{equation*}
    Multiplication with $ \frac{s_j}{s(\OPT \setminus G_{j-1})}$ and since $ s(\OPT \setminus G_{j-1}) \leq \capacity - s(\OPT \cap G_{j-1}) \leq \capacity $, we get
    \begin{align*}
        f(i_j \mid G_{j-1}) \geq \frac{c\,s_j}{\capacity} \Bigl( f(\OPT) - f(G_{j-1})\Bigr) + \frac{(1-c)s_j}{\capacity - s(\OPT \cap G_{j-1})} \Bigl( f(\OPT) - \sum_{\ell \in Q_{j-1}} f(i_\ell \mid G_{\ell-1})\Bigr),
    \end{align*}
    which, by the definition of $\delta_{j-1}$, completes the proof for $ \omega_{j,j} = 0 $. 
    
    Now, we consider the cases where $\omega_{j,j} \neq 0$, i.e., we consider $ j \leq k $ and $ \delta_j < \delta_{j-1} $. Recall that we have $ i_j \in \OPT $ under these assumptions. We can simply use $f(i_j \mid G_{j-1}) = s_j \delta_{j-1} + (f(i_j \mid G_{j-1}) - s_j \delta_{j-1})$ to derive the statement of the lemma. We have
    \begin{align*}
        s_j \delta_{j-1} &= c \, s_j \delta_{j-1} + (1-c) s_j \delta_{j-1} \\
        &= c\,s_j \frac{f(\OPT)-\sum_{\ell \in Q_{j-1}} f(i_\ell \mid G_{\ell-1})}{\capacity - s(\OPT \cap G_{j-1})} + (1-c)s_j \delta_{j-1} \\
        &\geq c\,s_j \frac{f(\OPT)-\sum_{\ell=1}^{j-1} f(i_\ell \mid G_{\ell-1})}{\capacity} + (1-c)s_j \delta_{j-1} \\
        &= \frac{c\,s_j}{\capacity} \Bigl(f(\OPT) - f(G_{j-1})\Bigr) + (1-c)s_j \delta_{j-1}.
    \end{align*}
    Thus, it remains to show that $ f(i_j \mid G_{j-1}) - s_j \delta_{j-1} \geq \omega_{j,j}$. By adding and subtracting $f(\OPT) - \sum_{\ell \in Q_{j-1}} f(i_\ell \mid G_{\ell-1})$, we get
    \begin{multline*}
        f(i_j \mid G_{j-1}) - s_j \delta_{j-1} \\ 
        = \Bigl(f(\OPT) - \sum_{\ell \in Q_{j-1}} f(i_\ell \mid G_{\ell-1})\Bigr) - s_j \delta_{j-1} - \Bigl(f(\OPT) -\sum_{\ell \in Q_{j}} f(i_\ell \mid G_{\ell-1})\Bigr),
    \end{multline*}
    since $i_j \in \OPT$ implies $j \in Q_j$. Using the definitions of $d_{j-1}$ and $d_{j}$, we can conclude that
    \begin{align*}
        f(i_j \mid G_{j-1}) - s_j \delta_{j-1} &= (\capacity - s(\OPT \cap G_{j-1})) \delta_{j-1} - s_j \delta_{j-1} - (\capacity - s(\OPT \cap G_j)) \delta_j \\
        &= (\capacity - s(\OPT \cap G_j)) (\delta_{j-1} - \delta_j) \\
        &\geq (\capacity - s(G_j)) (\delta_{j-1} - \delta_j) = \omega_{j,j},
    \end{align*}
    where we used again that $i_j \in \OPT$ and the inequality holds since $ \delta_j < \delta_{j-1} $.
\end{proof}

Before we sum up the marginal increases, we state a lemma that covers the behaviour of the $\omega$ terms within this summation. In order to sum over the $\omega$ terms that are nonzero, we define, for $j \in \{1,\dotsc,k\}$, the index set $D_j = \{\ell \in \{1,\dotsc,j\} \mid \delta_\ell < \delta_{\ell-1}\}$. Additionally, we use the notation $(x)^+ = \max\{x, 0\}$ for the positive part of a real number $x$.

\begin{lemma}
    \label{lem:omega-inequality}
    For each $ j \in \{2,\dotsc,k\} $, we have
    \begin{equation*}
        \biggl(1 - \frac{c\, s_j}{\capacity}\biggr) \sum_{\ell=1}^{j-1} \omega_{\ell,j-1} \geq (1-c) s_j (\delta_0 - \delta_{j-1})^+ + \sum_{\ell=1}^{j-1} \omega_{\ell,j}
    \end{equation*}
\end{lemma}

\begin{proof}
    Let $ j \in \{2,\dotsc,k\} $ be arbitrary. By the definitions of $\omega_{\ell,j-1}$ and $\omega_{\ell,j}$ for $\ell \in \{1,\dotsc,j-1\}$, we have
    \begin{align*}
        \biggl(1 - \frac{c\, s_j}{\capacity}\biggr) \sum_{\ell = 1}^{j-1} \omega_{\ell,j-1} &= \sum_{\ell \in D_{j-1}} \Biggl(\prod_{m=\ell+1}^{j} \biggl(1-\frac{c s_m}{\capacity}\biggr)\Biggr) (\capacity - s(G_{j-1}) (\delta_{\ell-1} - \delta_\ell) \\
        &= \sum_{\ell \in D_{j-1}} \omega_{\ell,j} + \sum_{\ell \in D_{j-1}} \Biggl(\prod_{m=\ell+1}^{j} \biggl(1-\frac{c s_m}{\capacity}\biggr)\Biggr) s_j (\delta_{\ell-1} - \delta_\ell).
    \end{align*}
    In order to simplify the product, we can apply a Weierstrass product inequality which states that for real numbers $ x_j \in [0,1], j \in \{1,\dotsc,n\} $, it holds that $\prod_{m=1}^n (1-x_m) \geq 1-\sum_{m=1}^n x_m$. Therefore, we have
    \begin{equation*}
        \prod_{m=\ell+1}^{j} \biggl(1-\frac{c \, s_m}{\capacity}\biggr) \geq 1 - \sum_{m=\ell+1}^{j} \frac{c \, s_m}{\capacity} = 1 - c \, \frac{\sum_{m=\ell+1}^{j} s_m}{\capacity} \geq 1 - c,
    \end{equation*}
    by the definition of $k$ and since $j \leq k$. Finally, we have 
    \begin{equation*}
        \sum_{\ell \in D_{j-1}} (\delta_{\ell-1}-\delta_{\ell}) \geq \max\Biggl\{\sum_{\ell=1}^{j-1} (\delta_{\ell-1}-\delta_\ell), 0\Biggr\} = (\delta_0 - \delta_{j-1})^+,
    \end{equation*}
    which completes the proof.
\end{proof}

The following lemma bounds the value of every prefix of the greedy sequence $f(G_j)$ in terms of $f(\OPT)$. For the proof, we use inductive arguments together with \Cref{lem:greedy-increase} and \Cref{lem:omega-inequality}.

\begin{lemma}
    \label{lem:greedy-sum}
    For each $ j \in \{1,\dotsc,k\} $, we have
    \begin{equation*}
        f(G_j) \geq \frac{1}{c} \Biggl( 1- \prod_{\ell=1}^{j} \biggl(1-\frac{c\, s_\ell}{\capacity}\biggr) \Biggr) f(\OPT) + \sum_{\ell=1}^{j} \omega_{\ell,j}.
    \end{equation*}
\end{lemma}

\begin{proof}
    We show the result by induction over $j$. For the base case $ j=1 $, we need to show that $ f(G_1) \geq \frac{s_1}{\capacity} f(\OPT) +  \omega_{1,1}$, which follows directly from \Cref{lem:greedy-increase} with $j=1$.
    
    Assume that the statement of the lemma holds up to $j-1$ and consider the statement for $j$. We have
    \begin{align*}
        f(G_j) &= f(G_{j-1}) + f(i_j \mid G_{j-1}) \\
        & \geq f(G_{j-1}) + \frac{c\, s_j}{\capacity} \Bigl(f(\OPT) - f(G_{j-1}) \Bigr) + (1-c)s_j \delta_{j-1} + \omega_{j,j} \\
        & = \biggl(1 - \frac{c\, s_j}{\capacity}\biggr)f(G_{j-1}) + \frac{c\, s_j}{\capacity} f(\OPT) + (1-c)s_j \delta_{j-1} + \omega_{j,j},
    \end{align*}
    where we used \Cref{lem:greedy-increase} for the inequality. Applying the induction hypothesis yields
    \begin{multline*}
        \biggl(1 - \frac{c\, s_j}{\capacity}\biggr)f(G_{j-1}) \\
        \begin{aligned}
        &\geq \biggl(1 - \frac{c\, s_j}{\capacity}\biggr) \Biggr( \frac{1}{c} \Biggl( 1- \prod_{\ell=1}^{j-1} \biggl(1-\frac{c\, s_\ell}{\capacity}\biggr) \Biggr) f(\OPT) + \sum_{\ell=1}^{j-1} \omega_{\ell,j-1} \Biggl) \\
        &= \frac{1}{c} \Biggl( \biggl(1 - \frac{c\, s_j}{\capacity}\biggr) - \prod_{\ell=1}^{j} \biggl(1-\frac{c\, s_\ell}{\capacity}\biggr) \Biggr) f(\OPT) + \biggl(1 - \frac{c\, s_j}{\capacity}\biggr) \sum_{\ell=1}^{j-1} \omega_{\ell,j-1} \\
        &= \frac{1}{c} \Biggl( 1 - \prod_{\ell=1}^{j} \biggl(1-\frac{c\, s_\ell}{\capacity}\biggr) \Biggr) f(\OPT) - \frac{s_j}{\capacity} f(\OPT) + \biggl(1 - \frac{c\, s_j}{\capacity}\biggr) \sum_{\ell=1}^{j-1} \omega_{\ell,j-1} \\
        &\geq \frac{1}{c} \Biggl( 1 - \prod_{\ell=1}^{j} \biggl(1-\frac{c\, s_\ell}{\capacity}\biggr) \Biggr) f(\OPT) - \frac{s_j}{\capacity} f(\OPT) + (1-c) s_j (\delta_0 - \delta_{j-1})^+ + \sum_{\ell = 1}^{j-1} \omega_{\ell,j},
        \end{aligned}
    \end{multline*}
    where we used \Cref{lem:omega-inequality} for the last inequality. Putting everything together gives
    \begin{multline*}
        f(G_j) \geq \frac{1}{c} \Biggl( 1 - \prod_{\ell=1}^{j} \biggl(1-\frac{c\, s_\ell}{\capacity}\biggr) \Biggr) f(\OPT) + \sum_{\ell=1}^j \omega_{\ell,j} \\
        -(1-c) \frac{s_j}{\capacity} f(\OPT) + (1-c)s_j \delta_{j-1} + (1-c) s_j (\delta_0 - \delta_{j-1})^+.
    \end{multline*}
    Since $\delta_0 = f(\OPT)/\capacity$, the latter part can be written as
    \begin{align*}
        (1-c) s_j (\delta_0 - \delta_{j-1})^+ - (1-c) s_j (\delta_0 - \delta_{j-1}),
    \end{align*}
    which is non-negative and thus, concludes the induction step and the proof of the lemma.
\end{proof}

In the following lemma, we simplify the result of \Cref{lem:greedy-sum} for $j = k$.

\begin{lemma}
    \label{lem:greedy-k-simplified}
    We have
    \begin{equation*}
        f(G_k) \geq \frac{1}{c} \Biggl(1- \exp\biggl(-c \, \frac{s(G_k)}{\capacity} \biggr) \Biggr) f(\OPT) + \biggl(1-c\,\frac{s(G_k)}{\capacity}\biggr)(\capacity-s(G_k)) (\delta_0 - \delta_{k})^+.
    \end{equation*}
\end{lemma}

\begin{proof}
    We get from \Cref{lem:greedy-sum} for $ j=k $ that
    \begin{equation}
        \label{eq:lem-greedy-k-simplified}
        f(G_k) \geq \frac{1}{c} \Biggl( 1 - \prod_{\ell=1}^{k} \biggl(1-\frac{c\, s_\ell}{\capacity}\biggr) \Biggr) f(\OPT) + \sum_{\ell=1}^{k} \omega_{\ell,k}.
    \end{equation}
    In order to obtain the statement of the lemma, we first bound the term in front of $f(\OPT)$ and afterwards, we bound the sum. We start by using an inequality that states that the geometric mean is always smaller or equal to the arithmetic mean for non-negative values $x_\ell$, $\ell \in \{1,\dotsc,k\}$, i.e., 
    \begin{equation*}
        \Biggl(\prod_{\ell=1}^k x_\ell\Biggr)^{\!1/k} \leq \frac{1}{k} \sum_{\ell=1}^k x_\ell.
    \end{equation*}
    We get
    \begin{align*}
        \frac{1}{c} \Biggl( 1- \prod_{\ell=1}^{k} \biggl(1-\frac{c\, s_\ell}{\capacity}\biggr) \Biggr) &\geq \frac{1}{c} \left[ 1- \Biggl(\frac{1}{k} \sum_{\ell=1}^{k} \biggl(1-\frac{c\, s_\ell}{\capacity}\biggr) \Biggr)^{\!k} \right] \\
        &= \frac{1}{c} \left[ 1- \Biggl(1 - \frac{c\, s(G_k)}{k \, \capacity}  \Biggr)^{\!k} \right] f(\OPT) \\
        & \geq \frac{1}{c} \Biggl(1- \exp\biggl(-c \, \frac{s(G_k)}{\capacity} \biggr) \Biggr) f(\OPT),
    \end{align*}
    where we first used that $ \sum_{\ell=1}^k s_\ell = s(G_k) $ and at the end, we applied the inequality $ (1+x/k)^k \leq e^x $ which holds for all $ k \geq 1 $ and $ x \in \mathbb{R} $. 
    
    Now, we consider the sum from the back of \eqref{eq:lem-greedy-k-simplified}. By the definitions, we have
    \begin{align*}
        \sum_{\ell=1}^{k} \omega_{\ell,k} &= \sum_{\ell \in D_k} \Biggl(\prod_{m=\ell+1}^{k} \biggl(1-\frac{c \, s_m}{\capacity}\biggr)\Biggr) (\capacity - s(G_k)) (\delta_{\ell-1} - \delta_\ell). 
    \end{align*}
    We first bound the product such that it becomes independent of $\ell$. By the Weierstrass product inequality (see proof of \Cref{lem:omega-inequality}), we get for each $\ell \in \{1,\dotsc,k\}$ that
    \begin{equation*}
        \prod_{m=\ell+1}^{k} \biggl(1-\frac{c \, s_m}{\capacity}\biggr) \geq 1 - \sum_{m=\ell+1}^{k}  \frac{c \, s_m}{\capacity} \geq 1-c\,\frac{s(G_k)}{\capacity},
    \end{equation*}
    Finally, we have
    \begin{equation*}
        \sum_{\ell \in D_k} (\delta_{\ell-1} - \delta_\ell) \geq \max \Biggl\{ \sum_{\ell=1}^{k} (\delta_{\ell-1} - \delta_\ell), 0 \Biggr\} = (\delta_0 - \delta_{k})^+,
    \end{equation*}
    which completes the proof.
\end{proof}

Now, we have all preliminary results together to obtain the approximation guarantee of \textsc{AGreedy} with the dependency on the curvature $c$. In order to obtain the guarantee, we combine \Cref{lem:greedy-k-simplified} with the statement of \Cref{lem:greedy-increase} for $j = k+1$ that we have not used yet.

\begin{theorem}
    \label{theo:greedy-approximation}
    Let $ c \in (0,1] $. For \textsc{AGreedy} we have
    \begin{equation*}
        f(\AG) \geq \frac{1}{c} \bigl(1- e^{-c x} \bigr) f(\OPT),
    \end{equation*}
    where $x$ is the unique root of $ \frac{1}{c} \bigl(1- e^{-c z} \bigr) = \frac{1-z}{1+c(1-z)} $ for $z \in [0,1]. $
\end{theorem}

\begin{proof}
    We let $ z = s(G_k)/\capacity $ be the fraction of the capacity that is used by the first $k$ items and we use the shorthand $ y = (\delta_0 - \delta_{k})^+$ for this proof. With this notation  \Cref{lem:greedy-k-simplified} yields
    \begin{equation}
        f(G_k) \geq \frac{1}{c} (1- e^{-cz} ) f(\OPT) + (1-cz)(1-z)\capacity y. \label{eq:theo-greedy-bound-1}
    \end{equation}
    Additionally, we have by \Cref{lem:greedy-increase} with $ j = k+1 $ that
    \begin{align*}
        f(i_{k+1} \mid G_k) &\geq \frac{c\, s_{k+1}}{\capacity} \Bigl(f(\OPT) - f(G_{k}) \Bigr) + (1-c)s_{k+1} \delta_{k} \\
        &> c (1-z) (f(\OPT) - f(G_{k})) + (1-c)(1-z) \capacity \delta_k,
    \end{align*}
    where we used that $ s_{k+1} > \capacity (1-z)$, since $ \capacity < s(G_k) + s_{k+1} = z\capacity + s_{k+1} $. By rearranging the terms and with $\delta_0 = f(\OPT)/\capacity$, we get
    \begin{align}
    \begin{split}
        f(i_{k+1} \mid G_k) + c (1-z) f(G_k) &> c (1-z) f(\OPT) + (1-c)(1-z) \capacity (\delta_k+\delta_0-\delta_0) \\
        &= (1-z) f(\OPT) - (1-c)(1-z) \capacity (\delta_0 - \delta_k) \\
        &\geq (1-z) f(\OPT) - (1-c)(1-z) \capacity y,
    \end{split}
    \label{eq:theo-greedy-bound-2}
    \end{align}
    where we used that $ (\delta_0 - \delta_k) \leq (\delta_0 - \delta_k)^+ = y $ for the last inequality. Recall that \textsc{AGreedy} returns $G_k$ if $ f(G_k) \geq f(i_{k+1} \mid G_k)$ and otherwise it returns $i_{k+1}$. Thus, we have that $ f(\AG) \geq \max\{f(G_k),f(i_{k+1} \mid G_k)\} $ and we get from \eqref{eq:theo-greedy-bound-2} that
    \begin{align*}
        (1 + c (1-z)) f(\AG) > (1-z) f(\OPT) - (1-c)(1-z) \capacity y.
    \end{align*}
    In summary, we have
    \begin{align*}
        f(\AG) &\geq \frac{1}{c} (1- e^{-c z}) f(\OPT) + (1-c z)(1-z)\capacity y =: g(y,z),\\
        f(\AG) &\geq \frac{1-z}{1 + c (1-z)} f(\OPT) - \frac{(1-c)(1-z)}{1 + c (1-z)} \capacity y =: h(y,z),
    \end{align*}
    where $y \geq 0$ and $z \in (0,1]$. Note that the first bound implies the statement of the lemma if $y \geq f(\OPT)/2\capacity$ since
    \begin{equation*}
        f(\AG) \geq \frac{1}{c} (1- e^{-c z}) f(\OPT) + \frac{1}{2} (1-c z)(1-z) f(\OPT) \geq \frac{1}{2} f(\OPT),
    \end{equation*}
    where the last inequality holds for every $c \in (0,1]$, since the left hand side is monotonically increasing for $z \in (0,1]$ and $ f(\OPT)/2 $ is obtained for $z=0$. The monotonicity holds, since we have for the derivative with respect to $z$ that
    \begin{align*}
        \Bigr(e^{-cz} - \frac{1}{2} (1-cz) - \frac{1}{2} c(1-z)\Bigl) f(\OPT) &\geq \Bigr(1-cz - \frac{1}{2} (1-cz) - \frac{1}{2} c(1-z)\Bigl) f(\OPT) \\
        &= \frac{1}{2} (1-c) f(\OPT) \geq 0,
    \end{align*}
    where we used that $e^x \geq 1+x$ for $x \in \mathbb{R}$.
    
    As a consequence, we want to find the value of
    \begin{equation}
        \label{eq:theo-greedy-approx-optimization}
        \min_{y \in [0,f(\OPT)/2\capacity),\,z \in (0,1]} \max\{g(y,z),h(y,z)\}.
    \end{equation}
    In the following, we first show that for a fixed value of $y$ the maximum of \eqref{eq:theo-greedy-approx-optimization} is always attained at the intersection of both functions and defines a value $z(y)$ with $g(y,z(y))=h(y,z(y))$. We avoid to calculate $z(y)$ explicitly, instead we obtain the derivative of $z$ via the implicit function theorem. Then, we can show that the value of $g$ is strictly increasing along the curve $(y,z(y))$ and thus, the minimum of \eqref{eq:theo-greedy-approx-optimization} is attained for $y=0$ which corresponds to the statement of the lemma. In \Cref{fig:graphs-of-g-and-h} we show plots of the functions $g$ and $h$ for different values of $c$ and $y$ that indicate the described behaviour of the curve along the intersections.
      
    Therefore, consider a fixed $y \in [0,f(\OPT)/2\capacity)$. We have that $g$ is strictly increasing and $h$ is strictly decreasing for $ z \in (0,1] $, since
    \begin{align}
    \begin{split}
        \label{eq:theo-greedy-g-increasing}
        \frac{\partial}{\partial z} g(y,z) &= e^{-cz} f(\OPT) - ((1-cz) + c(1-z)) \capacity y \\
        &\geq (1-cz) f(\OPT) - ((1-cz) + c(1-z)) \capacity y \\
        &> (1-cz) f(\OPT) - \frac{1}{2}\Bigl((1-cz) + c(1-z)\Bigr) f(\OPT)\\
        &= \frac{1}{2} (1-c) f(\OPT) \\
        &\geq 0,
    \end{split}
    \end{align}
    where we used that $e^x \geq 1+x$ for $x \in \mathbb{R}$ and that $y < f(\OPT)/2\capacity$. Furthermore,
    \begin{align}
    \begin{split}
        \label{eq:theo-greedy-h-decreasing}
        \frac{\partial}{\partial z} h(y,z) &= -\frac{1}{(1+c(1-z))^2} f(\OPT) + \frac{1-c}{(1+c(1-z))^2} \capacity y \\
        &< -\frac{1}{(1+c(1-z))^2} f(\OPT) + \frac{1-c}{2 (1+c(1-z))^2} f(\OPT) \\
        &= - \frac{1+c}{2 (1+c(1-z))^2} f(\OPT) \leq 0,
    \end{split}
    \end{align}
    where we used again that $y < f(\OPT)/2\capacity$. We conclude that both functions have a unique intersection $z(y) \in (0,1)$ for each $y \in [0,f(\OPT)/2\capacity)$, since
    \begin{align*}
        h(y,0) &= \frac{1}{1+c} f(\OPT) - \frac{1-c}{1+c} \capacity y > \frac{2}{1+c} \capacity y - \frac{1-c}{1+c} \capacity y = \capacity y = g(y,0), \\
        g(y,1) &= \frac{1}{c} \Bigl(1- e^{-c} \Bigr) f(\OPT) > 0 = h(y,1).
    \end{align*}
    Note that $g-h$ is continuously differentiable in the domains we consider for $y$ and $z$ and that the partial derivative with respect to $z$ is non-zero by \eqref{eq:theo-greedy-g-increasing} and \eqref{eq:theo-greedy-h-decreasing}. Therefore, we can apply the implicit function theorem on $g-h$, stating that $z(y)$ is a differentiable function with
    \begin{equation*}
        \frac{\mathrm{d}}{\mathrm{d}y}z(y) = - \Biggl( \frac{\partial}{\partial y} \Bigl(g(y,z(y)) - h(y,z(y))\Bigr) \Biggr) \Biggl( \frac{\partial}{\partial z} \Bigl(g(y,z(y)) - h(y,z(y)) \Bigr) \Biggr)^{-1}.
    \end{equation*}
    We want to show that
    \begin{equation*}
        \frac{\mathrm{d}}{\mathrm{d} y}g(y,z(y)) = \frac{\partial}{\partial y}g(y,z(y)) + \frac{\partial}{\partial z}g(y,z(y)) \frac{\mathrm{d}}{\mathrm{d}y}z(y) \geq 0,
    \end{equation*}
    implying the statement of the theorem. Rearranging the inequality after plugging in the derivative of $z(y)$ yields
    \begin{equation*}
        \frac{\frac{\partial}{\partial z}g(y,z(y)) - \frac{\partial}{\partial z}h(y,z(y))}{\frac{\partial}{\partial z}g(y,z(y))} \geq \frac{\frac{\partial}{\partial y}g(y,z(y)) - \frac{\partial}{\partial y}h(y,z(y))}{\frac{\partial}{\partial y}g(y,z(y))},
    \end{equation*}
    where we used that
    \begin{align}
    \begin{split}
        \label{eq:theo-greedy-y-derivatives}
        \frac{\partial}{\partial y} g(y,z) &= (1-c z)(1-z)\capacity \geq 0,\\
        \frac{\partial}{\partial y} h(y,z) &= - \frac{(1-c)(1-z)}{1 + c (1-z)} \capacity \leq 0.
    \end{split}
    \end{align}
    Further simplifying the inequality leads to
    \begin{equation*}
        \frac{\frac{\partial}{\partial z}h(y,z(y))}{\frac{\partial}{\partial z}g(y,z(y))} \leq \frac{\frac{\partial}{\partial y}h(y,z(y))}{\frac{\partial}{\partial y}g(y,z(y))}.
    \end{equation*}
    Using the inequalities of the partial derivatives from \eqref{eq:theo-greedy-g-increasing}, \eqref{eq:theo-greedy-h-decreasing} allows us to conclude, since
    \begin{align*}
        \dfrac{\frac{\partial}{\partial z}h(y,z(y))}{\frac{\partial}{\partial z}g(y,z(y))} &< - \frac{1+c}{(1-c)(1+c(1-z(y)))^2} \\
        & \leq - \frac{1}{(1-c)(1+c(1-z(y)))} \\
        & \leq - \frac{1-c}{(1-cz(y))(1+c(1-z(y)))} = \dfrac{\frac{\partial}{\partial y}h(y,z(y))}{\frac{\partial}{\partial y}g(y,z(y))},
    \end{align*}
    where we used for the second inequality that $ 1+c(1-z) \leq 1+c $ and for the last inequality that $ (1-c)^2 = 1 - c(2-c) \leq (1-cz) $.
\end{proof}

\begin{figure}[t]
\centering
\begin{tikzpicture}[scale=6]
    \clip (-0.95,-0.22) rectangle (1.45,0.85);
    \draw [-latex,very thick](0,-0.02)--(0,0.78);
    \draw [-latex,very thick](-0.02,0)--(0.56,0) node[right]{$z$};
    \foreach \x in {0.1, 0.2, 0.3, 0.4, 0.5}{
        \draw (\x,0.01 ) -- (\x,-0.01) node[anchor=north] {$\x$};
        }
    \foreach \x in {0.1, 0.2, 0.3, 0.4, 0.5, 0.6, 0.7}{
        \draw (0.01,\x) -- (-0.01,\x) node[anchor=east] {$\x$};
        }
    \draw[thick,mycolor1] plot[smooth] coordinates {(-0.1,0.475) (0,0.5) (0.1, 0.525) (0.2, 0.55) (0.3, 0.576) (0.4, 0.6025) (0.5, 0.63) (0.6,0.658)};
    \draw[thick,mycolor2] plot[smooth] coordinates {(-0.1,0.186) (0,0.25) (0.1, 0.311) (0.2, 0.37) (0.3, 0.427) (0.4, 0.483) (0.5, 0.536) (0.6,0.588)};
    \draw[thick,mycolor3] plot[smooth] coordinates {(-0.1,-0.102) (0,0) (0.1, 0.0975) (0.2, 0.190) (0.3, 0.279) (0.4, 0.363) (0.5, 0.442) (0.6,0.518)};
    
    \draw[thick,mycolor1] plot[smooth] coordinates {(-0.1,0.532) (0,0.5) (0.1, 0.4655) (0.2, 0.4286) (0.3, 0.3889) (0.4, 0.346) (0.5, 0.3) (0.6,0.25)};
    \draw[thick,mycolor2] plot[smooth] coordinates {(-0.1,0.621) (0,0.5833) (0.1, 0.5431) (0.2, 0.5) (0.3, 0.4537) (0.4, 0.404) (0.5, 0.35) (0.6,0.2917)};
    \draw[thick,mycolor3] plot[smooth] coordinates {(-0.1,0.71) (0,0.667) (0.1,0.621) (0.2, 0.5714) (0.3, 0.5185) (0.4, 0.4615) (0.5, 0.4) (0.6, 0.333)};
    \node[] at (0.3,-0.15) {$c=0.5$};
    \draw[mycolor1, fill] (0,0.5) circle[radius=.015cm];
    \draw[mycolor2, fill] (0.325,0.4415) circle[radius=.015cm];
    \draw[mycolor3, fill] (0.470,0.4189) circle[radius=.015cm];

\begin{scope}[shift={(0.8,0)}]
    \draw [-latex,very thick](0,-0.02)--(0,0.78);
    \draw [-latex,very thick](-0.02,0)--(0.56,0) node[right]{$z$};
    \foreach \x in {0.1, 0.2, 0.3, 0.4, 0.5}{
        \draw (\x,0.01 ) -- (\x,-0.01) node[anchor=north] {$\x$};
        }
    \foreach \x in {0.1, 0.2, 0.3, 0.4, 0.5, 0.6, 0.7}{
        \draw (0.01,\x) -- (-0.01,\x) node[anchor=east] {$\x$};
        }
    \draw[thick,mycolor1] plot[smooth] coordinates {(-0.1,0.4998) (0,0.5) (0.1, 0.5002) (0.2, 0.5013) (0.3, 0.5042) (0.4, 0.51) (0.5, 0.518) (0.6,0.5312)};
    \draw[thick,mycolor2] plot[smooth] coordinates {(-0.1,0.197) (0,0.25) (0.1, 0.298) (0.2, 0.341) (0.3, 0.382) (0.4, 0.42) (0.5, 0.456) (0.6,0.4912)};
    \draw[thick,mycolor3] plot[smooth] coordinates {(-0.1,-0.105) (0,0) (0.1, 0.0952) (0.2, 0.181) (0.3, 0.2592) (0.4, 0.3297) (0.5, 0.3935) (0.6,0.4512)};
    
    \draw[thick,blue!40!black] plot[smooth] coordinates {(-0.1,0.524) (0,0.5) (0.1, 0.474) (0.2, 0.444) (0.3, 0.412) (0.4, 0.375) (0.5, 0.333) (0.6,0.2857)};
    \node[] at (0.3,-0.15) {$c=1$};
    \draw[mycolor1, fill] (0,0.5) circle[radius=.015cm];
    \draw[mycolor2, fill] (0.3406,0.3973) circle[radius=.015cm];
    \draw[mycolor3, fill] (0.4428,0.3578) circle[radius=.015cm];
\end{scope}

\begin{scope}[shift={(-0.8,0)}]
    \draw [-latex,very thick](0,-0.02)--(0,0.78);
    \draw [-latex,very thick](-0.02,0)--(0.56,0) node[right]{$z$};
    \foreach \x in {0.1, 0.2, 0.3, 0.4, 0.5}{
        \draw (\x,0.01 ) -- (\x,-0.01) node[anchor=north] {$\x$};
        }
    \foreach \x in {0.1, 0.2, 0.3, 0.4, 0.5, 0.6, 0.7}{
        \draw (0.01,\x) -- (-0.01,\x) node[anchor=east] {$\x$};
        }
    \draw[thick,mycolor1] plot[smooth] coordinates {(-0.1,0.45) (0,0.5) (0.1, 0.55) (0.2, 0.6) (0.3, 0.65) (0.4, 0.7) (0.5, 0.75) (0.6,0.8)};
    \draw[thick,mycolor2] plot[smooth] coordinates {(-0.1,0.175) (0,0.25) (0.1, 0.325) (0.2, 0.4) (0.3, 0.475) (0.4, 0.55) (0.5, 0.625) (0.6,0.7)};
    \draw[thick,mycolor3] plot[smooth] coordinates {(-0.1,-0.1) (0,0) (0.1, 0.1) (0.2, 0.2) (0.3, 0.3) (0.4, 0.4) (0.5, 0.5) (0.6,0.6)};
    
    \draw[thick,mycolor1] plot[smooth] coordinates {(-0.1,0.55) (0,0.5) (0.1, 0.45) (0.2, 0.4) (0.3, 0.35) (0.4, 0.3) (0.5, 0.25) (0.6,0.2)};
    \draw[thick,mycolor2] plot[smooth] coordinates {(-0.1,0.825) (0,0.75) (0.1, 0.675) (0.2, 0.6) (0.3, 0.525) (0.4, 0.45) (0.5, 0.375) (0.6,0.3)};
    \draw[thick,mycolor3] plot[smooth] coordinates {(-0.1,1.1) (0,1) (0.1,0.9) (0.2, 0.8) (0.3, 0.7) (0.4, 0.6) (0.5, 0.5) (0.6, 0.4)};
    \node[] at (0.3,-0.15) {$c \rightarrow 0$};
    \draw[mycolor1, fill] (0,0.5) circle[radius=.015cm];
    \draw[mycolor2, fill] (0.333,0.5) circle[radius=.015cm];
    \draw[mycolor3, fill] (0.5,0.5) circle[radius=.015cm];
\end{scope}

    \draw[black] (1,0.62) rectangle (1.4,0.84);
    \draw[mycolor3] (1.04,0.80) -- (1.1,0.80);
    \node[anchor=west] at (1.1,0.8) {$y=0$};
    \draw[mycolor2] (1.04,0.73) -- (1.1,0.73);
    \node[anchor=west] at (1.1,0.73) {$y=0.25$};
    \draw[mycolor1] (1.04,0.66) -- (1.1,0.66);
    \node[anchor=west] at (1.1,0.66) {$y=0.5$};
    
\end{tikzpicture}
\caption{Visualization of the functions $g$ and $h$ with $f(\OPT)=\capacity=1$ and for different curvatures $c \in \{0,0.5,1\}$. We fixed values of $y$ and graphs with the same color belong to the same value of $y$. The increasing graphs are function $g$ and the decreasing graphs are function $h$; for $c=1$ the three graphs of $h$ are identical. The horizontal axis $z$ represents the total size of $G_k$ and the vertical axis the approximation guarantee. We can see that minimum of the maximum of $g$ and $h$ (the dots at the intersection) get lower if we decrease the value of $y$.}
\label{fig:graphs-of-g-and-h}
\end{figure}

The result of \Cref{theo:greedy-approximation} coincides with the known approximation guarantee of \textsc{MGreedy} for the cases $ c = 0 $ and $ c = 1 $. We have for the limit $ c \rightarrow 0 $ that 
\begin{equation*}
    \lim_{c \rightarrow 0} \frac{1}{c} \Bigl(1- e^{-c z} \Bigr) = z,
\end{equation*}
and the equation for the root simplifies to $ z = 1-z $ which implies the known approximation guarantee of $ 1/2 $ for the additive case; see, e.g., the textbook by Korte and Vygen~\cite{KorteV18}. For the other case of $ c = 1 $ the equation simplifies to $ 1-e^{-z} = \frac{1-z}{2-z} $, which implies the approximation guarantee shown by Wolsey~\cite{DBLP:journals/mor/Wolsey82}.

In the remainder of the chapter, we give an instance of the knapsack problem for a submodular set function with curvature $ c \in (0,1] $ where \textsc{MGreedy} attains the approximation guarantee of \Cref{theo:greedy-approximation}. Together with \Cref{prop:greedy-algorithms} this proves that the approximation guarantee is tight for \textsc{AGreedy} and \textsc{MGreedy}.

Consider a capacity $ \capacity = 1 $ and a set of items $ N = \{i_0,i_1,\dotsc,i_k,i_{k+1}\} $ with $ k \in \mathbb{N} $. We assign each item $ i \in N $ a value $ v(i) \in \mathbb{R}_{> 0} $ and for $ S \subseteq N $ the value of the submodular set function $ f $ is given by  
\begin{equation*}
    f(S) = \begin{cases}
        \sum_{i \in S} v(i) & \text{ if } i_0 \notin S, \\
        v(i_0) + (1-c) \sum_{i \in S \setminus \{i_0\}} v(i) & \text{ if } i_0 \in S \text{ and } i_{k+1} \notin S, \\
        v(i_0) + v(i_{k+1}) + (1-c) \sum_{i \in S \setminus \{i_0,i_{k+1}\}} v(i) & \text{ if } i_0 \in S \text{ and } i_{k+1} \in S, \\
    \end{cases}
\end{equation*}
where $ c \in (0,1] $. Using \eqref{def:Subm1} as the definition of submodularity makes it easy to verify that $ f $ is submodular and it is also easy to see that $ f $ is normalized and monotone. In terms of curvature, we have for every item $ i \in N \setminus \{i_0\} $ that $f(i \mid N \setminus \{i\}) \geq (1-c) v(i) = (1-c) f(i)$. Moreover, for $i_0 \in N$, we have to choose the values such that
\begin{equation}
    \label{eq:curvature-inequality}
    f(i_0 \mid N \setminus \{i_0\}) = v(i_0) - c \sum_{i \in N \setminus \{i_0,i_{k+1}\}} v(i) \geq (1-c) v(i_0) = (1-c) f(i_0).
\end{equation}
For the values, we let $ z $ be a real number in $ (0,1) $ that we specify later and we let $ s = z/k $. For $ j \in \{1,\dotsc,k\} $, we set
\begin{align*}
    v(i_0) &= 1, & v(i_j) &= s(1-c\,s)^{j-1}, & v(i_{k+1}) &= (1-z+\varepsilon)\biggl(1-c\,\sum_{\ell=1}^{k} v(i_\ell)\biggr), \\
    s(i_0) &= 1, & s(i_j) &= s, & s(i_{k+1}) &= 1-z+\varepsilon,
\end{align*}
with a sufficiently small $ \varepsilon > 0 $. Note that the inequality in \eqref{eq:curvature-inequality} can be written as $ v(i_0) \geq \sum_{\ell=1}^k v(i_\ell) $, and is fulfilled since
\begin{equation*}
    \sum_{\ell=1}^k v(i_\ell) = s \sum_{\ell=1}^k (1-cs)^{\ell-1} \leq s \, k = z \leq 1 = v(i_0),
\end{equation*}
and thus, $f$ has curvature $c$.

\textsc{MGreedy} builds the following solution: in the first iteration the algorithm chooses the item maximizing $ v(i)/s(i) $. This results in a tie between $ i_0 $ and $ i_1 $ and we decide, as an adversary, that the algorithm chooses $ i_1 $. Assume for a later iterations $ j \in \{2,\dotsc,k+1\} $, that the greedy solution up to iteration $j-1$ is $ G_{j-1} = \{i_1,\dotsc,i_{j-1}\} $. The algorithm chooses the item maximizing $ f(i \mid G_{j-1})/s(i) $ from the remaining items and this results for each $ j \in \{2,\dotsc,k+1\} $ in a tie between $ i_0 $ and $ i_j $, since we have
\begin{align*}
    \frac{f(i_0 \mid G_{j-1})}{s(i_0)} = v(i_0)-c \sum_{\ell=1}^{j-1} v(i_\ell) &= 1-c\,s \sum_{\ell=1}^{j-1} (1-c\,s)^{\ell-1} \\
    &= 1-c\,s \, \frac{1-(1-c\,s)^{j-1}}{c\,s} = (1-c\,s)^{j-1},
\end{align*}
and, for $j \in \{2,\dotsc,k\}$, we have
\begin{align*}
    \frac{f(i_j \mid G_{j-1})}{s(i_j)} = \frac{v(i_j)}{s(i_j)} = (1-c\,s)^{j-1},
\end{align*}
and lastly, for $j = k+1$, we have
\begin{align*}
    \frac{f(i_{k+1} \mid G_{k})}{s(i_{k+1})} = \frac{v(i_{k+1})}{s(i_{k+1})} = 1-c\,\sum_{\ell=1}^{k} v(i_\ell) &= 1-c\,s \sum_{\ell=1}^{k} (1-c\,s)^{\ell-1} \\
    &= 1-c\,s \, \frac{1-(1-c\,s)^{k}}{c\,s} = (1-c\,s)^{k}.
\end{align*}
We decide again that the algorithm chooses item $ i_{j} $ in iteration $ j \in \{2,\dotsc,k+1\} $. With this inductive argumentation, we know that \textsc{MGreedy} returns either $ G_k = \{i_1,\dotsc,i_k\} $ or $\{i_{k+1}\}$, because $ i_{k+1} $ is the first item that exceeds the capacity. On the other hand, the optimal solution would be $\OPT = \{i_0\}$ with a value of $1$. The value of $ G_k $ is 
\begin{align*}
    f(G_k) = \sum_{m=1}^{k} s(1-c\,s)^{m-1} = s \frac{1-(1-c\,s)^k}{c\,s} &= \frac{1}{c} \Bigl(1-(1-c\,s)^k\Bigr)\\ 
    &= \frac{1}{c} \biggl(1-\Bigl(1-\frac{c\,z}{k}\Bigr)^k\biggr) \\
    &\xrightarrow{k\rightarrow\infty} \frac{1}{c} \Bigl(1-e^{-cz}\Bigr),
\end{align*}
and the value of $ i_{k+1} $ is 
\begin{equation*}
    f(i_{k+1}) = v(i_{k+1}) = (1-z+\varepsilon)\biggl(1-c\sum_{m=1}^{k} v(i_m)\biggr) = (1-z+\varepsilon)(1-c \, f(G_k)).
\end{equation*}
For a sufficiently small $\varepsilon > 0 $, we can choose $ z \in (0,1)$ such that $ f(G_k) = f(i_{k+1}) $ and this yields 
\begin{equation*}
    f(i_{k+1}) = \frac{1-z+\varepsilon}{1+c(1-z+\varepsilon)}.
\end{equation*}
Thus, we have $f(\MG) = \max\{f(G_k), f(i_{k+1})\} = \frac{1}{c} \bigl(1- e^{-c x} \bigr)$ 
where $x$ is the unique root of the equation $ \frac{1}{c} \bigl(1- e^{-c z} \bigr) = \frac{1-z+\varepsilon}{1+c(1-z+\varepsilon)} $ for $ z \in (0,1) $. With $ \varepsilon \rightarrow 0 $, we get arbitrary close to the result of \Cref{theo:greedy-approximation}.

\section{Submodular Knapsack Problem with Unknown Capacity}
\label{sec:robust}

In this section we introduce an algorithm that generates a policy that is always at least as good as \textsc{AGreedy} even though it does not know the capacity of the knapsack. For that purpose we introduce indispensable items in the first part of this section. They are defined similar to swap items defined by Disser et al.~\cite{DBLP:journals/siamdm/DisserKMS17}, which they used to achieve their $ 1/2 $-optimal policy for an additive objective function.

As discussed by Kawase et al.~\cite{DBLP:journals/siamdm/KawaseSF19}, one major challenge when going from the case of an additive objective function to a submodular objective function is that the greedy order of items depends on the capacity of the knapsack. When an item is packed into the knapsack then other items that have a large overlap in terms of the objective with the packed item decrease in density. On the other hand, for another capacity where the first item is not packed since it does not fit they remain attractive. This issue makes it difficult to compare the outcome of a packing policy that does not know the capacity with the outcome of the \textsc{MGreedy} algorithm as it was done in Disser et al.~\cite{DBLP:journals/siamdm/DisserKMS17}.

Kawase et al.~\cite{DBLP:journals/siamdm/KawaseSF19} overcome this issue by introducing the concept of a \emph{single-valuable item} $ i $ with the property $ f(\{i\}) \geq 2f(\OPT(s(i)/2)) $, i.e., Kawase et al. do not compare items with the greedy solution at all and instead compare the value of an item directly with the optimal solution $\OPT$. In their policy, the most valuable single-valuable item that fits in the knapsack is inserted first. Afterwards, they try to insert the rest of the items in their greedy order. This deterministic policy achieves a robustness factor of $2(1-1/e)/21 \approx 0.06$.

To motivate the usage of the alternative greedy algorithm \textsc{AGreedy}, we show in the following example that it is not possible to find policies that are always as good as \textsc{MGreedy}. Consider items $ i_1,i_2,i_3,$ and $i_4 $ with sizes
\begin{equation*}
    s(i_1) = 5, \quad s(i_2) = 5, \quad s(i_3) = 12, \quad s(i_4) = 10.
\end{equation*}
They are labeled in the order in that \textsc{MGreedy} would consider them. The submodular set function $ f $ is defined by the set coverage in \Cref{fig:impossibleMG}, where elements are represented by positive numbers and the value of a subset of items is given by the sum of all elements covered by those items. We show that no matter with which item the policy starts there is a capacity for that \textsc{MGreedy} returns a more valuable solution than the policy.

Assume the policy starts with either $ i_1 $ or $ i_2 $. For capacity $ \capacity = 11 $ the solution of the policy will be $ \{i_1,i_2\} $ with $ f(\{i_1,i_2\}) = 19 $, whereas \textsc{MGreedy} returns $ \{i_4\} $ with $ f(i_4) = 20 $. This is because $ i_3 $ gets discarded at the beginning of the algorithm and $ i_4 $ becomes the first item that exceeds the capacity.

If the policy starts with $ i_3 $, then for $ \capacity = 12 $ the solution of the policy is $ \{i_3\} $ with $ f(i_3) = 12 $. However, the output of \textsc{MGreedy} for this capacity is $ \{i_1,i_2\} $ with $ f(\{i_1,i_2\}) = 19 $.

Now, assume the policy start with $ i_4 $ and the capacity is $ \capacity = 22 $. Depending on the second item of the policy the solution is either $ \{i_4,i_1,i_2\} $ with a value of $ 27 $ or $ \{i_4,i_3\} $ with a value of $ 28 $, whereas \textsc{MGreedy} returns $ \{i_1,i_2,i_3\} $ with a value of $ 31 $.

We overcome the capacity-dependency of the greedy order and the impossibility result for \textsc{MGreedy} by defining the concept of \emph{indispensable items}. These are items that the alternative greedy algorithm \textsc{AGreedy} returns instead of the greedy solution. It turns out that we can find a policy that is as good as \textsc{AGreedy} for every capacity. The difference between \textsc{MGreedy} and \textsc{AGreedy} in the previously shown example lies in the first case with $\capacity=11$. Here, \textsc{AGreedy} returns the greedy solution instead of $i_4$, since the marginal increase of the item with respect to $\{i_1,i_2\}$ is smaller than the value of $\{i_1,i_2\}$.

\begin{figure}
    \begin{center}
        \begin{tikzpicture}
            \draw[draw=black,thick, fill = lightgray] (0,0) -- (2,0) -- (2,2) -- (0,2) -- cycle;
            \draw[draw=black,thick] (1.3,0.1) -- (1.3,1.9) -- (2.7,1.9) -- (2.7,0.1) -- cycle;
            \draw[draw=black,thick] (-0.7,1.05) -- (0.7,1.05) -- (0.7,1.9) -- (-0.7,1.9) -- cycle;
            \draw[draw=black,thick] (-0.7,0.95) -- (0.7,0.95) -- (0.7,0.1) -- (-0.7,0.1) -- cycle;
            \node[] () at (-0.35,0.5) {$ 3 $};
            \node[] () at (-0.35,1.5) {$ 4 $};
            \node[] () at (0.35,0.5) {$ 4 $};
            \node[] () at (0.35,1.5) {$ 8 $};
            \node[] () at (1,1) {$ 4 $};
            \node[] () at (1.65,1) {$ 4 $};
            \node[] () at (2.35,1) {$ 8 $};
            \node[] () at (1,2.3) {$ i_4 $};
            \node[] () at (-1,1.5) {$ i_1 $};
            \node[] () at (3.3,1) {$ i_3 $};
            \node[] () at (-1,0.5) {$ i_2 $};
        \end{tikzpicture}
        \caption{
            \label{fig:impossibleMG}
            Coverage function $ f $ with four items $ i_1,\dotsc,i_4 $.}
    \end{center}
\end{figure}

\subsection{Indispensable Items}
\label{subsec:indis}

\begin{definition}
    An item $ i \in \N $ is called \emph{indispensable} if there exists a capacity $\capacity \in \mathbb{R}_{>0} $ for that \textsc{AGreedy} returns $\{i\} =\{i_{k+1}\}$ instead of the greedy solution $G_k$.
\end{definition}
For a fixed capacity $\capacity \in \mathbb{R}_{>0} $, we say that item $ i \in \N $ is indispensable for $\capacity $ if \textsc{AGreedy} returns $\{i\} = \{i_{k+1}\}$ instead of the greedy solution $ G_k $.

For ease of exposition, we assume in the following that there are no ties when an algorithm compares items by value, differences in value, or density. In practice this could be achieved by small perturbations of the values, or by using a lexicographic order that breaks ties in a systematic way. However, to avoid heavy notation, we assume that ties do not exist.

In the following, we will denote by $ U = \{i \in N \mid s(i) \leq \capacity\} $ the set of items considered by \textsc{AGreedy}. Since the capacity is unknown to us, we might have to deal with a larger set of items. Therefore let $ \theta \in \mathbb{R}_{>0} $ be an arbitrary threshold on the item sizes. We let $ N_{\theta} = \{i \in N \mid s(i) \leq \theta \} $ denote the subset of all items with sizes not larger than $ \theta $ and we define $ n_\theta = \left|N_\theta\right| $.

The following lemma contains important properties that will help us to determine indispensable items which are a key element of our robust policy.
 
\begin{lemma}	
    \label{lem:indispensable-items}
    Let $ \theta \in \mathbb{R}_{>0} $ and let $ i_1, i_2, \dotsc i_{n_\theta} $ be the items in $ N_{\theta} $ sorted by their greedy order. If there exists a natural number $k \geq 1 $, such that $ f(i_{k+1} \mid G_k) > f(G_k)$, the following properties hold:
    \begin{enumerate}[(i)]
        \item \label{it:lem-indispensible-1} $ s(i_{k+1}) > \sum_{j=1}^{k} s(i_j) $.
        \item \label{it:lem-indispensible-2} $ i_{k+1} $ is an indispensable item for capacity $s(i_{k+1})$
        \item \label{it:lem-indispensible-3} If it exists, let $ \tilde{\theta} $ be the smallest capacity larger than $ s(i_{k+1}) $, such that the first $ k+1 $ items in the greedy order of $ N_{\tilde{\theta}} $ are not the first $ k+1 $ items in the greedy order of $ N_{s(i_{k+1})} $. Then, the first item in the greedy order of $ N_{\tilde{\theta}} $ that is not identical to the item in the greedy order of $ N_{\theta} $ is either the first item in the greedy order of $ N_{\tilde{\theta}} $ or an indispensable item for capacity $ \tilde{\theta} $.
    \end{enumerate}
\end{lemma}

\begin{proof} 
    We start by showing property \eqref{it:lem-indispensible-1}. By the assumption of the lemma and by the submodularity of $f$, we have for all $ j \in \{1,\dotsc,k\} $ that
    \begin{align}
        \label{eq:L1}
        f(G_k) < f(i_{k+1} \mid G_k) \leq f(i_{k+1} \mid G_{j-1}).
    \end{align} 
    Additionally, by the definition of the greedy order, we have for each $ j \in \{1,2,\dotsc,k\}$ that
    \begin{equation}
        \label{eq:L2}
        \frac{f(i_j \mid G_{j-1})}{s(i_j)} \geq \frac{f(i_{k+1} \mid G_{j-1})}{s(i_{k+1})},
    \end{equation} 
    Therefore, we have
    \begin{align*}
        f(G_k) &= \sum_{j=1}^{k} f(i_j \mid G_{j-1}) 
        \geq \sum_{j=1}^{k} \frac{s(i_j)}{s(i_{k+1})} f(i_{k+1} \mid G_{j-1}),
    \end{align*}
    where we used \eqref{eq:L2} for the inequality. We obtain
    \begin{equation*}
        s(i_{k+1}) \geq \sum_{j=1}^{k} s(i_j) \frac{f(i_{k+1} \mid G_{j-1})}{f(G_k)} > \sum_{j=1}^{k} s(i_j)
    \end{equation*}
    by inequality~\eqref{eq:L1}.
    
    We continue with property \eqref{it:lem-indispensible-2}. Let $ U = \{i \in N \mid s(i) \leq s(i_{k+1})\} $ be the set of items considered by \textsc{AGreedy} for capacity $s(i_{k+1})$. By property \eqref{it:lem-indispensible-1}, we know that the items $i_1,\dotsc,i_{k+1}$ are contained in $U$ and therefore, they are considered by \textsc{AGreedy} in the same order for capacity $s(i_{k+1})$. Again by property \eqref{it:lem-indispensible-1}, we get that the first $k$ items don't exceed the capacity and thus, $i_{k+1}$ is the first item that exceeds the capacity. Since $ f(i_{k+1} \mid G_k) > f(G_k)$ holds by assumption, \textsc{AGreedy} returns item $i_{k+1}$ instead of the greedy solution which is what we needed to show.
    
    Finally, we show \eqref{it:lem-indispensible-3}.
    Let $ j_1,\dotsc,j_{k+1} $ be the first $ k+1 $ items in the greedy order of $ N_{\tilde{\theta}} $ and, for $ m \in \{1,\dotsc,k+1\} $, let $ H_m = \bigcup_{\ell=1,\dots,m} \{j_\ell\} $. Further, let $ \smash{\tilde{k}} $ be the largest index such that $ i_\ell = j_\ell $ for all $ \smash{\ell \in \{1,\dotsc,\tilde{k}\}} $. By the definitions of $\smash{\tilde{\theta}}$ and $\smash{\tilde{k}}$, we know that $ \smash{\tilde{k} \leq k} $, $ G_m = H_m $ for all $ m \in \{1,\dotsc,\tilde{k}\} $, and $ s(j_{\tilde{k}+1}) = \tilde{\theta} $.
    
    We proceed to prove that $ j_{\tilde{k}+1} $ is an indispensable item for $ \smash{\tilde{\theta}} $, if $ \smash{\tilde{k} \geq 1} $. This implies the statement, since $\smash{\tilde{k} = 0}$ implies that it is the first item of the greedy order of $N_{\tilde{\theta}}$.
    Since $ s(H_{\tilde{k}}) = s(G_{\tilde{k}}) < s(i_{k+1}) < s(j_{\tilde{k}+1}) $, we have that $ j_{\tilde{k}+1} $ is the first item in the greedy order of $ N_{\tilde{\theta}} $ that exceeds the capacity. Additionally, we have
    \begin{equation*}
        f(j_{\tilde{k}+1} \mid H_{\tilde{k}}) \geq f(i_{k+1} \mid H_{\tilde{k}}) \geq f(i_{k+1} \mid G_k) > f(G_k) \geq f(H_{\tilde{k}}).
    \end{equation*}
    The first inequality holds, since $ j_{\tilde{k}+1} $ is in front of $ i_{k+1} $ in the greedy order $ G_{\tilde{\theta}} $ and $ s(j_{\tilde{k}+1}) > s(i_{k+1}) $. Second and last inequality follow from submodularity, since $ H_{\tilde{k}} = G_{\tilde{k}} \subseteq G_k $. The third inequality holds by the assumption of the lemma. Therefore, $ j_{\tilde{k}+1} $ is returned by \textsc{AGreedy} instead of the greedy solution and thus, an indispensable item for capacity $ \smash{\tilde{\theta}} $ if $ \tilde{k} \geq 1 $.
\end{proof}	

\subsection{A Robust Policy}
\label{subsec:robust}

    \begin{algorithm}[tb]
    \caption{Construction of a starting policy}
    \label{alg:starting-policy}
    \begin{algorithmic}[1]
        \Procedure{Policy}{$ N $}
        \State $ G_0 \gets \emptyset; k \gets 0; n \gets |N| $
        \For{$ j \gets 1,\dotsc,n $}
        \State $ i_j \gets \arg\max_{i \in N} \Bigl\{\frac{f(i \,|\, G_{j-1})}{s(i)}\Bigr\} $
        \If{$ j \geq 2 \textbf{ and } f(i_j \mid G_{j-1}) > f(G_{j-1}) $}
        \State $k \gets j-1$
        \EndIf
        \State $ G_j \gets G_{j-1} \cup \{i_j\} $
        \State $ N \gets N \setminus \{i_j\} $
        \EndFor
        \State \Return $ \Pi \gets (i_{k+1},i_1,\dotsc,i_k,i_{k+2},\dotsc,i_n)$
        \EndProcedure
    \end{algorithmic}
\end{algorithm}

The general idea of the adaptive policy is to choose a reasonable start item based on \Cref{lem:indispensable-items}~(\ref{it:lem-indispensible-3}). By the statement it is convenient to start with an indispensable item as long as there are no larger indispensable items and the first item in the greedy order is not larger than the indispensable item. With that in mind, \Cref{alg:starting-policy} creates a first policy by building the greedy order of all items. Meanwhile, it looks for the last item $i_j$ with $ j \geq 2 $ in the greedy order that fulfills the condition $f(i_j \mid G_{j-1}) > f(G_{j-1})$. If such an item exists, it is swapped to the front of policy while all other items stay in their greedy order. In the following lemma, we show some properties of the first item in the policy provided by \Cref{alg:starting-policy}.

\begin{lemma}
    \label{lem:properties-starting-policy}
    Let $ \theta \in \mathbb{R}_{>0} $ and let $ i_1, i_2, \dotsc i_{n_\theta} $ be the items in $ N_{\theta} $ sorted by their greedy order. Moreover, let $\Pi$ be the policy returned by \Cref{alg:starting-policy} with $ N_{\theta} $ given as the input and we define $k$ such that $i_{k+1}$ is the first item of $\Pi$. Then, the following properties hold:
    \begin{enumerate}[(i)]
        \item \label{it:policy-lemma-1} If $k = 0$, then $i_{k+1}$ is larger than all indispensable items of $ N_{\theta} $.
        \item \label{it:policy-lemma-2} If $k \geq 1$, then $i_{k+1}$ is the largest indispensable item of $ N_{\theta} $.
    \end{enumerate}
\end{lemma}

\begin{proof}
    If $ k \geq 1 $, then $i_{k+1}$ fulfills $ f(i_{k+1} \mid G_{k}) > f(G_k) $. Therefore, by \Cref{lem:indispensable-items}~(\ref{it:lem-indispensible-2}), we know that $i_{k+1}$ is an indispensable item and, by \Cref{lem:indispensable-items}~(\ref{it:lem-indispensible-1}), we know that all items $ i_1,\dotsc,i_k $ are smaller than $i_{k+1}$. For $k\geq0$, assume there are larger indispensable items than $ i_{k+1} $. Let $ i_j $ with $ j > k+1 $ be the largest of those indispensable items. Then there has to be some capacity $ \omega $ for which $i_j$ is indispensable. Note that indispensable items naturally fulfill the assumptions of \Cref{lem:indispensable-items}. We will show that such a capacity $ \omega $ cannot exist.
    
    First of all, we have $ f(i_j \mid G_{j-1}) \leq f(G_{j-1}) $ since otherwise $ i_{k+1} $ would not be the first item in $ \Pi $. Therefore, $ i_j $ is not indispensable for capacity $ \theta $. For all capacities larger than $ \theta $ an item needs to change the greedy order in front of $ i_j $ before it can become indispensable. But then there is always an item in front of $ i_j $ that is larger than $ i_j $, contradicting \Cref{lem:indispensable-items}~(\ref{it:lem-indispensible-1}). 
    If $ \omega $ is smaller than $ \theta $ it is clear that the greedy order of $ N_{\omega} $ has to change in front of $ i_j $ for some capacity $ \tilde{\omega} $ with $ \omega < \tilde{\omega} < \theta $. But then, by \Cref{lem:indispensable-items}~(\ref{it:lem-indispensible-3}), there is either an indispensable item larger than $ i_j $, immediately contradicting our assumption that $ i_j $ is the largest, or there is an item at the beginning of the new greedy order that is larger than $ i_j $ and thus, larger than $ i_{k+1} $. If $ k=0 $ this is a contradiction since $i_{k+1}$ cannot be the first item in the greedy order of $ N_\theta $ anymore and if $ k \geq 1 $, then \Cref{lem:indispensable-items}~(\ref{it:lem-indispensible-1}) is violated for capacity $ \theta $ and therefore, $i_{k+1}$ does not fulfil $ f(i_{k+1} \mid G_{k}) > f(G_k) $, again a contradiction
\end{proof}

\begin{algorithm}[tb]
    \caption{Packing and adapting of the policy}
    \label{alg:adaptive-policy}
    \begin{algorithmic}[1]
        \State $ S \gets \emptyset $
        \While{$ S = \emptyset $}
        \State $ \Pi \gets \textsc{Policy}(N) $
        \If{$ \Pi(1) $ fits in the knapsack}
        \State $ S \gets S \cup \{ \Pi(1) \}; N \gets N \setminus \{ \Pi(1) \} $
        \Else
        \State $ \N \gets \{i \in \N \mid s(i) < s(\Pi(1))\} $
        \EndIf
        \EndWhile
        \State $j \gets 2$
        \While{$N \neq \emptyset$}
        \If{$ \Pi(j) $ fits in the knapsack}
        \State $ S \gets S \cup \{\Pi(j)\};  N \gets N \setminus \{ \Pi(j) \} $
        \State $j \gets j+1$
        \Else
        \State $ \N \gets \{i \in \N \mid s(i) < s(\Pi(j))\} $
        \State $ \Pi \gets \textsc{Update\_Policy}(N,S) $
        \State $j \gets 1$
        \EndIf
        \EndWhile
        \State \Return S
    \end{algorithmic}
\end{algorithm}

The adaptive policy in \Cref{alg:adaptive-policy} consists of two steps. In the first step the algorithm begins with all items in $ N $ and starts packing items in the order provided by \Cref{alg:starting-policy}. If the first item $i_{k+1}$ of the order can be added to the knapsack, i.e., $ s(i_{k+1}) \leq \capacity $, the algorithm continues to pack items in the predefined order. Otherwise it discards $i_{k+1}$ together with all items that are at least as large as $i_{k+1}$. Then, Algorithm~\ref{alg:starting-policy} provides a new order of the remaining items $ N_\theta $ with $ \theta = s(i_{k+1}) $ and the process is repeated until the first item of the order can be added to the knapsack.

In the second part, after the first item is added, the algorithm tries to add items in the present order. If an item is added to the knapsack the algorithm just continues with the next item of the order. Only if an item is not added, we have to adapt our policy. This is necessary, because $N_\theta$ can contain items with a size larger than the capacity and those items can change the greedy order compared to the greedy order considered by \textsc{AGreedy}. Therefore, \Cref{alg:policy-update} updates the greedy order of the remaining items based on the set if items $S$ that we have already packed into the knapsack.

\begin{algorithm}[tb]
    \caption{Update of the greedy order}
    \label{alg:policy-update}
    \begin{algorithmic}[1]
        \Procedure{Update\_Policy}{$ N,S $}
        \State $ G_0 \gets S;  n \gets |N| $
        \For{$ j \gets 1,\dotsc,n $}
        \State $ i_j \gets \arg\max_{i \in N} \Bigl\{\frac{f(i \,|\, G_{j-1})}{s(i)}\Bigr\} $
        \State $ G_j \gets G_{j-1} \cup \{i_j\} $
        \State $ N \gets N \setminus \{i_j\} $
        \EndFor
        \State \Return $ \Pi \gets (i_1,\dotsc,i_n)$
        \EndProcedure
    \end{algorithmic}
\end{algorithm}

Finally, we want to show that the packing obtained from the policy of \Cref{alg:adaptive-policy} is as good as \textsc{AGreedy} for all capacities. The prove is based on the following insight: if the packing started with an indispensable item $ i_{k+1}$ with $ k \geq 1 $ and one of the following items $i_1,\dotsc,i_k$ cannot be added, then \textsc{AGreedy} returned the same indispensable item. If all of these items can be added to the knapsack, we are back on packing items in their greedy order and \textsc{AGreedy} returns the greedy solution as we do. Should we start with the first item of the greedy order, i.e., $k=0$, we can show that \textsc{AGreedy} returns the greedy solution as well.

\begin{theorem}
    \label{theo:agreedy-vs-policy}
    For a capacity $ \capacity \in \mathbb{R}_{>0} $ let $ S(\capacity) $ be the output of Algorithm~\ref{alg:adaptive-policy} and let $ \AG(\capacity) $ be the output of \textsc{AGreedy}. Then, we have $f(S(\capacity)) \geq f(\AG(\capacity))$ for every capacity $ \capacity \in \mathbb{R}_{>0}$.
\end{theorem}

\begin{proof}
    Consider an arbitrary capacity $ \capacity \in \mathbb{R}_{>0} $. Let $N_{\theta}$ with $ \theta \geq \capacity $ be the set of items and let $\Pi$ be the policy at the time the first item of $\Pi$ is added to the solution $S$ in \Cref{alg:adaptive-policy}. Furthermore, let $i_1,\dotsc,i_{n_\theta}$ be the items in $N_{\theta}$ sorted by their greedy order and let $i_{k+1}, k \in \{0,\dotsc,n_\theta-1\}$ be the first item in the policy $\Pi$. By assumption we have $\capacity \geq s(i_{k+1})$. We distinguish the cases $k=0$ and $k \geq 1$.
    
    If $k \geq 1$, then the packing continues with the items $i_1,\dotsc,i_k$ and afterwards with the items $i_{k+1},\dotsc,i_{n_\theta}$. First of all, we assume that we cannot add all items $i_1,\dotsc,i_k$ to our knapsack, i.e., $s(\{i_1,\dotsc,i_{k}\}) + s(i_{k+1}) > \capacity $. By \Cref{lem:indispensable-items}~(\ref{it:lem-indispensible-1}) we get that $\capacity \geq s(i_{k+1}) > s(\{i_1,\dotsc,i_{k}\}) $ and therefore, we know that the greedy solution of \textsc{AGreedy} is $G_k = \{i_1,\dotsc,i_k\} $ and $ i_{k+1} $ is the first item that exceeds the capacity. Since $ i_{k+1} $ was swapped to the front of the policy in \Cref{alg:starting-policy}, we know that $ f(i_{k+1} \mid G_k) > f(G_k)$ and thus, \textsc{AGreedy} returns $ i_{k+1} $ instead of the greedy solution. We get $f(S(\capacity)) \geq f(i_{k+1}) = f(\AG(\capacity))$ as claimed.
    
    In case we have $k \geq 1$ and $ s(\{i_1,\dotsc,i_{k}\}) + s(i_{k+1}) \leq \capacity $, all items $ i_1,\dotsc,i_{k} $ are packed and together with $ i_{k+1} $ we have packed the first $ k+1 $ items of the greedy order. For $ k = 0 $ we are in the same situation after we packed the first item $ i_{k+1} $. The algorithm continues to pack items $ i_{k+2},\dotsc,i_{n_\theta} $. If we get to an item $ j $ with $ s(j) > \capacity $, it cannot be added and we update the greedy order of the remaining items by \Cref{alg:policy-update} such that we add items in the exact same order as they are added to the greedy solution of \textsc{AGreedy}. The first item $ j $ with $ s(j) \leq \capacity $ that cannot be added is also the first item that exceeds the capacity in \textsc{AGreedy}. If \textsc{AGreedy} returns the greedy solution $G$ we get $f(S(\capacity)) \geq f(G) = f(\AG(\capacity))$ as claimed. Otherwise, if $j$ is returned by \textsc{AGreedy}, $j$ is an indispensable item for capacity $\capacity$. But this cannot be, because $ s(j) \leq s(i_{k+1}) $ would contradict \Cref{lem:indispensable-items}~(\ref{it:lem-indispensible-1}) and $ s(j) > s(i_{k+1}) $ would contradict \Cref{lem:properties-starting-policy}~(\ref{it:policy-lemma-1}) if $ k=0 $ and \Cref{lem:properties-starting-policy}~(\ref{it:policy-lemma-2}) if $ k\geq 1 $.
\end{proof}

From \Cref{theo:greedy-approximation} and \Cref{theo:agreedy-vs-policy}, we obtain the main result of this paper.

\begin{theorem}
    There exists an adaptive policy that is $\alpha$-optimal for $ \alpha = \frac{1}{c} \bigl(1- e^{-c x} \bigr) $ where $x$ is the unique root in $[0,1]$ of the equation $ \frac{1}{c} \bigl(1- e^{-c z} \bigr) = \frac{1-z}{1+c(1-z)} $.
\end{theorem}

\bibliographystyle{abbrvnat}
\bibliography{Indis}
	
\end{document}